\documentclass[journal]{IEEEtran}

\usepackage[tbtags]{amsmath} % defines many math commands and
\usepackage{amssymb}  % get, among others, blackboard bold fonts
\usepackage{verbatim} % get comment environment, + new verbatim
\usepackage{amsxtra}  % get, eg, \accentedsymbol

\usepackage{enumerate}

\usepackage{amsfonts}
\usepackage{color}

\usepackage{subfig}

\usepackage{ mathrsfs}

\usepackage{cite}

\usepackage{mathtools}
\mathtoolsset{showonlyrefs=true}
\usepackage{hyperref}
\usepackage{units}

\usepackage{epsfig, graphicx, psfrag}

\usepackage[mathcal]{euscript}
\newenvironment{frcseries}{\fontfamily{frc}\selectfont}{}

\usepackage{ifthen}
\providecommand{\VersionLength}{long}
\newcommand{\ver}{\ifthenelse{\equal{\VersionLength}{long}}}
\newcommand{\nver}{\ifthenelse{\equal{\VersionLength}{short}}}

\usepackage{amsthm}
\theoremstyle{plain}

\newtheorem{thm}{Theorem}
\newtheorem{lemma}{Lemma}

\newtheorem{corol}{Corollary}
\newtheorem{prop}{Proposition}

\theoremstyle{definition}

\newtheorem{defn}{Definition}

\newtheorem*{scheme*}{Scheme}

\theoremstyle{remark}

\newtheorem{remark}{Remark}

\providecommand{\thmref}[1]{Theorem~\ref{#1}}

\providecommand{\defnref}[1]{Definition~\ref{#1}}
\providecommand{\secref}[1]{Section~\ref{#1}}
\providecommand{\lemref}[1]{Lemma~\ref{#1}}
\providecommand{\propref}[1]{Proposition~\ref{#1}}

\providecommand{\figref}[1]{Fig.~\ref{#1}}

\providecommand{\appref}[1]{Appendix~\ref{#1}}

\newcommand{\ie}{i.e.}
\newcommand{\eg}{e.g.}

\newcommand{\viz}{viz.}
\newcommand{\etal}{\emph{et al.}}

\newcommand{\bm}[1]{\mbox{\boldmath{$#1$}}}

\newcommand{\SINR}{\text{SINR}}

\newcommand{\eff}{{\text{eff}}}

\newcommand{\Comment}[1]{}
\newcommand{\old}[1]{}
\newcommand{\rem}[1]{}

\newcommand{\eps}{{\epsilon}}
\newcommand{\teps}{{\tilde{\eps}}}

\newcommand{\vphi}{\varphi}

\newcommand{\hx}{\hat{x}}

\providecommand{\tT}{\tilde T}

\newcommand{\ty}{\tilde y}

\providecommand{\tx}{\tilde x}
\providecommand{\tz}{\tilde z}

\newcommand{\tR}{\tilde R}

\newcommand{\bD}{{\bf D}}
\newcommand{\bt}{\bm t}

\newcommand{\by}{\bm y}

\newcommand{\bH}{\text{\bf H}}
\newcommand{\bI}{\text{\bf I}}
\newcommand{\bT}{\text{\bf T}}

\newcommand{\bK}{\text{\bf K}}
\newcommand{\bB}{\text{\bf B}}
\newcommand{\bQ}{\text{\bf Q}}
\newcommand{\bU}{\text{\bf U}}
\newcommand{\bV}{\text{\bf V}}

\providecommand{\bU}{\text{\bf U}}
\providecommand{\bV}{\text{\bf V}}

\providecommand{\tbU}{\tilde{\bU}}
\providecommand{\bD}{\text{\bf D}}

\newcommand{\bX}{{\bf X}}

\providecommand{\tbT}{\tilde{\bT}}
\newcommand{\bG}{\mathbf{G}}
\newcommand{\bA}{\mathbf{A}}
\providecommand{\tbG}{\tilde{\bG}}

\newcommand{\bL}{{\bf D}}
\newcommand{\bx}{{\bm x}}
\newcommand{\br}{{\bm r}}

\providecommand{\bu}{{\bm u}}

\newcommand{\bY}{{\bf Y}}

\providecommand{\bV}{{\bf V}}
\newcommand{\bz}{{\bm z}}

\providecommand{\be}{{\bm e}}

\providecommand{\tbL}{\tilde{\bL}}
\providecommand{\tL}{\tilde{D}}

\newcommand{\bzero}{\text{\bf 0}}

\providecommand{\comment}[1]{}

\newcommand{\beqn}[1]{\begin{eqnarray}\label{#1}}
\newcommand{\eeqn}{\end{eqnarray}}
\newcommand{\beq}[1]{\begin{equation}\label{#1}}
\newcommand{\eeq}{\end{equation}}

\newcommand{\tby}{\tilde {\bm y}}
\newcommand{\tbx}{\tilde {\bm x}}

\newcommand{\tbz}{\tilde {\bm z}}

\newcommand{\cC}{{\mathcal C}}
\providecommand{\rank}{\text{rank}}

\providecommand{\trace}{\text{trace}}

\providecommand{\trace}{\text{trace}}

\providecommand{\tby}{\tilde {\bm y}}
\providecommand{\tbx}{\tilde {\bm x}}

\providecommand{\tbz}{\tilde {\bm z}}

\makeatletter
\newcommand{\vast}{\bBigg@{4}}
\newcommand{\Vast}{\bBigg@{5}}
\makeatother 

\providecommand{\Na}{{N_A}}
\providecommand{\Ne}{{N_E}}
\providecommand{\Nc}{{N_C}}
\providecommand{\Nb}{{N_B}}

\providecommand{\Lb}{{L_B}}
\providecommand{\Lc}{{L_C}}
\providecommand{\Le}{{L_E}}

\newcommand{\rvm}{{m}}
\newcommand{\rvf}{{f}}
\newcommand{\rvx}{{\tilde{\mathsf{x}}}}
\newcommand{\rvy}{{\mathsf y}}
\newcommand{\rvby}{{\mathbf y}}

\newcommand{\del}{{\delta}}
\newcommand{\rvu}{{\mathsf u}}
\newcommand{\rvv}{{\mathsf v}}
\newcommand{\rvz}{{\mathsf z}}
\newcommand{\defeq}{\triangleq}

\providecommand{\MI}[2]{I \left( #1 ; #2 \right)}
\providecommand{\CMI}[3]{I \left( #1 ; #2 \middle| #3 \right)}

\newcommand{\bsigma}{{\bm \sigma}}
\newcommand{\bmu}{{\bm \mu}}

\providecommand{\obG}{\bar{\bG}}
\providecommand{\obK}{\bar{\bK}}

\providecommand{\ubD}{\breve{\bD}}

\providecommand{\ubU}{\breve{\bU}}
\providecommand{\ubV}{\breve{\bV}}

\providecommand{\thx}{\hat{\tx}}

\begin{document}
%%%%%%%%%%%%%%%%%%%%%%%%%%%%%%%%%%%%%%%%%%%%%%%%%%%%%%%%%%%%%%%%%%%%%%%%%%%%%%%%%%%%%%%%%%%%%%%%%%%%%%%%%%%%%%%%%%%%%
\title{The MIMO Wiretap Channel Decomposed}

\author{
Anatoly Khina, 
Yuval Kochman, 
and 
Ashish Khisti
  \thanks{The material in this paper was presented in part at the \emph{2014 IEEE International Symposium of Information Theory (ISIT)}, Honolulu, HI, USA, 
          and at the \emph{2015 IEEE ISIT}, Hong Kong.
  }
}

% make the title area
\maketitle

%%%%%%%%%%%%%%%%%%%%%%%%%%%%%%%%%%%%%%%%%%%%%%%%%%%%%%%%%%%%%%%%%%%%%%%%%%%%%%%%%%%%%%%%%%%%%%%%%%%%%%%%%%%%%%%%%%%%%%%%%

\begin{abstract}
    The problem of sending a secret message over the Gaussian multiple-input multiple-output (MIMO) wiretap channel is studied.
    While the capacity of this channel is known, it is not clear how to construct optimal coding schemes that achieve this capacity.
    In this work, we use linear operations along with successive interference cancellation to attain effective parallel single-antenna wiretap channels.
    By using independent scalar Gaussian wiretap codebooks over the resulting parallel channels, the capacity of the MIMO wiretap channel is achieved.
    The derivation of the schemes is based upon joint triangularization of the channel matrices. 
    We find that the same technique can be used to re-derive capacity expressions for the MIMO wiretap channel in a way 
    that is simple and closely connected to a transmission scheme. 
    This technique allows to extend the 
    previously proven strong 
    security for scalar Gaussian channels to the MIMO case.
    We further consider the problem of transmitting confidential messages over a two-user broadcast MIMO channel. For that problem, we find that derivation of both the capacity and a transmission scheme is a direct corollary of the proposed analysis for the MIMO wiretap channel. 
\end{abstract}


\allowdisplaybreaks
%%%%%%%%%%%%%%%%%%%%%%%%%%%%%%%%%%%%%%%%%%%%%%%%%%%%%%%%%%%%%%%%%%%%%%%%%%%%%%%%%%%%%%%%%%%%%%%%%%%%%%%%%%%%%%%%%%%%%%%%%

\section{Introduction}

The wiretap channel, introduced by Wyner \cite{Wyner75_wiretap}, 
is composed of a sender (``Alice'') who wishes to convey data to a legitimate user (``Bob''), 
such that the eavesdropper (``Eve'') cannot recover (almost) any information of these data.
The capacity of this channel~\cite{Wyner75_wiretap,CsiszarKorner_confidentialBC} equals to a mutual-information difference, 
and was extended to the Gaussian case in \cite{GaussianWiretap}.
Let the channels from Alice to Bob and Eve be given by
\begin{align} 
    y_B &= h_B x + z_B ,
 \\ y_E &= h_E x + z_E , 
\end{align}
where $h_B$ and $h_E$ are complex scalar gains, $z_B$ and $z_E$ are mutually-independent 
circularly-symmetric Gaussian zero mean unit variance noises 
and the transmission is subject to a unit power constraint. Then, the capacity is achieved by a Gaussian input:
% \vspace{-\baselineskip}
\begin{subequations}
\label{eq:AWGN_wiretap}
\noeqref{eq:AWGN_wiretap:MI_diff,eq:AWGN_wiretap:Gaussian}
\begin{align} 
    C_S(h_B, h_E) 
    &= \MI{x}{y_B} - \MI{x}{y_E}
\label{eq:AWGN_wiretap:MI_diff}
 \\ &= \left[ \log \left( 1 + \left| h_B \right|^2 \right) 
      - \log \left( 1 + \left|h_E\right|^2 \right) \right]_+ \!, \ \ \ \ 
\label{eq:AWGN_wiretap:Gaussian}
\end{align}
\end{subequations}
where $[a]_+ \triangleq \max \{0, a\}$ is the positive-part operation.

The vector extension of this result, the multiple-input multiple-output (MIMO) Gaussian wiretap channel
or the multiple-input multiple-output multiple-eavesdropper (MIMOME) channel \cite{KhistiWornell_MIMO_wiretap,OggierHassibi_MIMO_wiretap,LiuShamai_MIMO_wiretap}, is given by
\begin{subequations}
\label{eq:MIMOME_channel}
  \noeqref{eq:Bob's_channel,eq:Eve's_channel}
    \begin{align} 
        \by_B &= \bH_B \bx + \bz_B,  
    \label{eq:Bob's_channel} 
     \\* \by_E &= \bH_E \bx + \bz_E, 
    \label{eq:Eve's_channel} 
    \end{align} 
\end{subequations} 
where $\bx$, $\by_B$ and $\by_E$ are complex-valued vectors with dimensions of the number of antennas in the terminals of Alice, Bob and Eve, 
denoted by $\Na$, $\Nb$, and $\Ne$, respectively. The channel matrices $\bH_B$ and $\bH_E$ have the corresponding dimensions. 
The additive noise vectors $\bz_B$ and $\bz_E$ are mutually independent, i.i.d., circularly-symmetric Gaussian with zero mean unit element variance.

The secrecy capacity of this scenario for the case where the input 
is subject to an average \emph{covariance constraint}\footnote{$\bA \succeq \bzero$ denotes that 
$\bA$ is a positive semidefinite matrix. 
$\bA \preceq \bB$ means that $(\bA - \bB) \succeq \bzero$.}
\begin{align}
\label{eq:covariance_constraint}
    \bK \triangleq E \left[ \bx \bx^\dagger \right] \preceq \obK ,
\end{align}
and the case where the input is subject to a total (over all antennas) power constraint $P$:
\begin{align}
\label{eq:tot_power_constraint}
    \trace( \bK ) \leq P
    ,
\end{align}
was established in \cite{LiuShamai_MIMO_wiretap} and \cite{KhistiWornell_MIMO_wiretap,OggierHassibi_MIMO_wiretap,LiuShamai_MIMO_wiretap}, respectively.
Under a covariance constraint, 
this capacity is given by the difference of \ver{mutual informations}{MIs} to Bob and Eve, 
optimized over all Gaussian channel inputs that satisfy the respective input constraint:
\begin{align} 
\label{eq:wiretap_covariance}
    C_S(\bH_B,\bH_E,\obK) = \max_{ \bK \preceq \obK} I_S(\bH_B,\bH_E,\bK) \,,
\end{align}
where 
\begin{align}
\label{eq:Gaussian_MI_difference}
    I_S(\bH_B,\bH_C,\bK) \triangleq I(\bH_B,\bK) - I(\bH_E,\bK),
\end{align}
and 
\begin{align} 
\label{eq:Gaussian_MI}
    I(\bH,\bK) \triangleq \log \left| \bI + \bH \bK \bH^\dagger \right| 
\end{align} 
is the Gaussian vector mutual information (MI), and $|\bA|$ denotes the determinant of $\bA$.
Later, Bustin \etal~\cite{Wiretap_BustinEURASIP} 
provided an explicit solution to the maximization problem under the covariance constraint~\eqref{eq:wiretap_covariance}. 
A closed-form solution for the wiretap capacity under a total power constraint is yet to be found, 
although a numerical algorithm that approaches the global optimum was recently proposed \cite{LoykaCharalambous:NumericalMIMO-WTC-capacity}. 
We note that the capacity under a total power constraint can be written as the union of achievable regions under a covariance constraint (see~\cite[Lemma~1]{WSS06}):
\begin{align}
\label{eq:wiretap_power}
    C_S(\bH_B, \bH_E, P) 
    &= \max_{\obK: \: \trace\{\obK\} = P} C_S(\bH_B,\bH_E,\obK) \,.
\end{align}
Hence, we shall concentrate on the covariance constrained setting in this paper.

The confidential broadcast channel offers a natural extension to the wiretap channel setting. In the confidential broadcast setting, 
Alice wishes to convey different data to two users (``Bob'' and ``Charlie''), 
such that (almost) no information can be recovered by one user about the data intended for the other user.
That is, for the data that are intended for Bob, Charlie acts as the eavesdropper (``Eve'' in the wiretap setting), whereas for the data intended for Charlie, Bob takes the role of Eve.

The capacity region of the Gaussian MIMO confidential broadcast channel, a scenario considered first in \cite{ConfidentialMIMO_BC_totalPower}, 
was determined by Liu \etal~\cite{ConfidentialMIMO_BC} to be rectangular
under the covariance constraint \eqref{eq:covariance_constraint}.
Namely, it is given by all rate pairs $\left( R_B, R_C \right)$ satisfying 
\begin{subequations} 
\label{eq:BC_C}
\noeqref{eq:BC_C:Bob,eq:BC_C:Charlie}
\begin{align} 
    R_B & \leq C_S \left( \bH_B,\bH_C,\obK \right) ,
\label{eq:BC_C:Bob}
 \\ R_C & \leq C_S \left( \bH_C,\bH_B,\obK \right) , 
\label{eq:BC_C:Charlie}
\end{align}
\end{subequations}
where $\bH_C$ is the channel matrix to Charlie replacing $\bH_E$ in \eqref{eq:Eve's_channel}, and
$C_S(\bH_B,\bH_C,\obK)$ is the capacity of the MIMO wiretap channel defined in \eqref{eq:wiretap_covariance}.
The converse is immediate, as both users achieve their maximal possible secrecy rates simultaneously; it is the direct part that is quite striking. 

Although capacity is well understood, it is less clear how to construct codes for wiretap and confidential broadcast channels. 
For the scalar Gaussian case, various approaches have been suggested, see, \eg, \cite{WiretapCodeExpanderISIT2014,LDPCforWTC_IT07,LDPCforWiretap2011,WiretapLatticeCodes_OggierSoleBelfiore,WiretapPolarCodes,WiretapPracticalCodingPhD,StrongSecrecyWiretapPolarLattice_ISIT2014} 
and references therein.
However, assuming that we have such a code for the scalar case, it is not clear how to construct a capacity-achieving scheme for the MIMO setting.

In this work we present an approach that reduces these MIMO secrecy problems to scalar Gaussian ones by means of matrix decompositions, specifically joint unitary triangularizations \cite{STUD:SP}. The decompositions yield a layered coding scheme, where the secrecy capacity is approached by means of a scalar wiretap code in each layer and successive interference cancellation (SIC) at the receiver.
The contribution of such an approach to the MIMO wiretap channel can be compared to that of singular-value decomposition (SVD) based schemes~\cite{Telatar99}, or Vertical Bell-Laboratories Space--Time (V-BLAST) and decision feedback equalization (GDFE) schemes \cite{Foschini96,Wolniansky_V-BLAST,CioffiForneyGDFE,HassibiVBLAST}, to MIMO communication without secrecy constraints. 

Beyond the architectural merit, our approach yields two more fruits. First, it enables us to revisit the capacity results for the MIMO wiretap and confidential MIMO broadcast channels.
In that respect, we establish the optimal covariance matrix for the MIMO wiretap channel as well as an expression for the secrecy capacity in terms of the generalized singular values of suitably defined matrices. 
This re-derives a result by Bustin~\etal~\cite{Wiretap_BustinEURASIP}, which was based on 
elaborate information-theoretic considerations,
using a direct linear-algebraic approach. Turning to the confidential broadcast channel, we are able to re-derive \eqref{eq:BC_C} almost as a corollary of the analysis applied to the MIMO wiretap channel, also explaining the role of dirty-paper coding in this setup.

Second, reducing the MIMO problem to a scalar one allows us to leverage recent advances in the secrecy analysis of the scalar Gaussian wiretap channel: 
whereas we concentrate in this paper on constructing \emph{weak secrecy} schemes, namely schemes for which 
\begin{align}
\label{eq:weak_secrecy}
    \MI{\bx^n}{\by_B^n} \leq n \eps, 
\end{align}
we show that in fact a special matrix triangularization allows to achieve \emph{strong secrecy} guarantees 
for the MIMO wiretap channel, \ie, 
\begin{align}
\label{eq:strong_secrecy}
    \MI{\bx^n}{\by_B^n} \leq \eps ,
\end{align}
where both \eqref{eq:weak_secrecy} and \eqref{eq:strong_secrecy} hold for any $\eps > 0$ and large enough blocklength $n$.

An outline of this paper is as follows. 
We start by reviewing the relevant unitary matrix decompositions in \secref{s:decompositions}.
These decompositions are used to re-derive the MIMO wiretap capacity expressions in \secref{s:CapacityRevisited}.
We further recall how these decompositions allow to construct capacity-achieving schemes for the MIMO channel without secrecy in \secref{ss:MIMO_P2P:SIC}.
We extend this framework to work for the MIMO wiretap setting in \secref{s:schemes}. 
Layered dirty-paper coding (DPC)~\cite{Costa83} variants of this scheme are discussed in \secref{s:DPC}
and are also shown to be capacity achieving.
Finally, these schemes are utilized, along with the results of \secref{s:CapacityRevisited}, to construct a simple 
proof of the capacity region of the confidential MIMO broadcast setting as well as providing a layered-DPC scheme that attains it in \secref{s:BC}.

%%%%%%%%%%%%%%%%%%%%%%%%%%%%%%%%%%%%%%%%%%%%%%%%%%%%%%%%%%%%%%%%%%%%%%%%%%%%%%%%%%%%%%%%%%%%%%%%%%%%%%%%%%%%%%%%%%%%%%%%%

\section{Unitary Matrix Triangularization}
\label{s:decompositions}

In this section we briefly review some important matrix decompositions which will be used in the sequel.
In \secref{ss:GTD} we recall the generalized triangular decomposition (GTD),
and some of its important special cases which include the SVD, QR decomposition, and geometric mean decomposition (GMD).\footnote{See \cite{JET:SeveralUsers2015:FullPaper} for a geometrical interpretation of these decompositions.}
Joint unitary triangularizations of two matrices are discussed in \secref{ss:STUD}.

Throughout this paper, we shall only need to decompose full-rank matrices with equal or more rows than columns.

%----------------------------------------------------------------------------------------------------------------

\subsection{Single Matrix Triangularization}
\label{ss:GTD}

The following definitions are used in this section.

\begin{defn}[Multiplicative majorization; see \cite{PalomarJiang}]
\label{def:major}
    Let $\bx$ and $\by$ be two $N$-dimensional vectors of positive elements.
    Denote by $\tbx$ and $\tby$ the vectors composed of the entries of $\bx$ and $\by$, respectively, ordered non-increasingly.
    We say that $\bx$ majorizes $\by$ ($\bx \succeq \by$) if they have equal products:
    \begin{align}
        \prod_{j=1}^N  x_j  =  \prod_{j=1}^N y_j  \,,
    \end{align}
    and their (ordered) elements satisfy,
    for any $1 \leq \ell < N$,
    \begin{align}
        \prod_{j=1}^\ell  \tx_j  \geq  \prod_{j=1}^\ell  \ty_j  \,.
    \end{align}
\end{defn}

\begin{defn}[Singular values; see~\cite{GolubVanLoan3rdEd}]
    Let $\bA$ be a full-rank matrix of dimensions $M \times N$, where $M \geq N$.
    Then, the singular values (SVs) of $\bA$ are the positive solutions $\sigma$ of the equation
    \begin{align}
        \left| \bA^\dagger \bA - \sigma^2 \bI \right| = 0.
    \end{align}
    Let the SV vector $\bsigma(\bA)$ be composed of all SVs (including their algebraic multiplicity), ordered non-increasingly.
\end{defn}

The following is a straightforward extension of the definition of triangular matrices to non-square ones.
\begin{defn}[Generalized Upper-Triangular Matrix]
\label{def:UpperTriangular}
    An \mbox{$M \times N$} matrix is said to be generalized upper triangular if 
    \begin{align}
	T_{i,j} &= 0   \,, & \forall i &> j \,;& i &= 1, \ldots, M \,;& j &= 1, \ldots, N .
    \end{align}
\end{defn}

We use these definitions to characterize the set of all possible diagonals achievable via unitary triangularization, as follows.

\begin{thm}[Generalized Triangular Decomposition]
\label{thm:GTD}
    Let $\bA$ be a full-rank matrix of dimensions $M \times N$, where $M \geq N$, and $\bt$ be an $N$-dimensional vector of positive
    elements.
    A GTD of the matrix $\bA$ is given by
        \begin{align} 
        \label{eq:GTD}
            \bA &= \bU \bT \bV^\dagger ,
        \end{align}
    where $\bU$ and $\bV$ are unitary matrices of dimensions $M \times M$ and $N \times N$, respectively, and $\bT$ is a generalized upper-triangular matrix
    with a prescribed set of diagonal values
    $\bt$, \ie, 
    \begin{align}
        T_{ii} &= t_i \,, & i = 1, \ldots, N \,,
    \\ T_{ij} &= 0   \,, & \forall i > j \,.
    \end{align}
    Such a decomposition
    exists if and only if the vector $\bt$
    is majorized by $\bsigma(\bA)$:
    \begin{align}
    \label{eq:majorization_gtd}
        \bsigma(\bA) \succeq \bt \,.
    \end{align}
\end{thm}
In other words, the singular values are an extremal case for the diagonal of all possible unitary triangularizations.

The necessity of the majorization condition was proven by Weyl \cite{WeylCondition}. 
Horn further showed that for any $\br$ that is majorized by $\bsigma$, 
there exists an upper triangular matrix with diagonal $\br$ and SV vector $\bsigma$~\cite{WeylConditionInverse_ByHorn}.
The sufficiency of the majorization condition as it appears in \thmref{thm:GTD} was proved in \cite{UnityTriangularization,GTD,QRS-GTD}, where also explicit constructions of the decomposition
were introduced.

We now recall three important special cases of the GTD.
\subsubsection{SVD (See, \eg, \cite{GolubVanLoan3rdEd})}
Here the resulting matrix $\bT$ in \eqref{eq:GTD} is a \emph{diagonal} matrix, 
and its diagonal elements are equal to the singular values of the decomposed matrix $\bA$.

\subsubsection{QR Decomposition (See, \eg, \cite{GolubVanLoan3rdEd})}
In this decomposition, 
the matrix
$\bV$ in \eqref{eq:GTD} equals to the identity matrix and hence does not depend on the matrix $\bA$.
This decomposition can be constructed by performing Gram--Schmidt orthonormalization on the (ordered) columns of the
matrix $\bA$.

\subsubsection{GMD (See \cite{UnityTriangularization,QRS,GMD})}
\label{sss:GMD}
The diagonal elements of $\bT$ in this decomposition are all equal 
to the geometric mean of its singular values $\bsigma(\bA)$, 
which is real and positive.
Note that this decomposition always exists if $\bA$ is full rank 
(since the vector of the SVs of $\bA$ necessarily majorizes the vector of the diagonal elements of $\bT$), but is not unique.

%------------------------------------------------------------------------------------------------------------------------

\subsection{Joint Matrix Triangularization}
\label{ss:STUD}

The existence condition for a joint unitary triangularization of two matrices is similar to that of the GTD in \thmref{thm:GTD}, 
where the singular values are replaced by the generalized singular values (GSVs), 
and the diagonal of $\bT$ is replaced by the ratio of the diagonals of the resulting generalized triangular matrices.
These quantities are defined below.

\begin{defn}[Generalized singular values \cite{VanLoan76,GolubVanLoan3rdEd}]
\label{def:GSV}
    For any (ordered) matrix pair $(\bA_1,\bA_2)$, the GSVs are the non-negative solutions $\mu$ of the equation
    \[
        \left| \bA_1^\dagger \bA_1 - \mu^2 \bA_2^\dagger \bA_2 \right| = 0 .
    \]
    Let the GSV vector $\bmu(\bA_1, \bA_2)$ be composed of all GSVs (including their algebraic multiplicity), ordered non-increasingly.
\end{defn}

A characterization of the possible joint unitary triangularizations of two matrices with prescribed diagonal ratios is provided in the following theorem.

\begin{thm}[Joint unitary triangularization \cite{STUD:SP}]
\label{thm:STUD}
    Let $\bA_1$ and $\bA_2$ be two full-rank matrices of dimensions $M_1 \times N$ and $M_2 \times N$, respectively, where $M_1, M_2 \geq N$, and $\bt$ be an $N$-dimensional vector of positive elements.
    A joint unitary triangularization of the matrices $\bA_1$ and $\bA_2$ is given by
    \begin{subequations}
    \label{eq:STUD}
    \noeqref{eq:STUD:A1,eq:STUD:A2}
    \begin{align} 
        \bA_1 &= \bU_1 \bT_1 \bV^\dagger ,
    \label{eq:STUD:A1}
     \\ \bA_2 &= \bU_2 \bT_2 \bV^\dagger ,
    \label{eq:STUD:A2}
    \end{align}
    \end{subequations}
    where $\bU_1$, $\bU_2$ and $\bV$ are unitary matrices of dimensions \mbox{$M_1 \times M_1$}, $M_2 \times M_2$ and $N \times N$, respectively, and $\bT_1$ and $\bT_2$ are generalized upper-triangular matrices (recall \defnref{def:UpperTriangular}) with a prescribed set of diagonal ratios 
    $\bt$, \ie, 
    \begin{align}
        \frac{T_{1;ii}}{T_{2;ii}} &= t_i \,, &&              && i = 1, \ldots, N \,,
     \\ T_{k;i,j} &= 0              \,, && k = 1, 2 \,, && \forall i > j \,.
    \end{align}
    Such a joint decomposition
    exists if and only if the vector $\bt$
    is majorized by the GSV vector $\bmu(\bA_1, \bA_2)$:
    \begin{align}
    \label{eq:majorization_gsvd}
        \bmu(\bA_1, \bA_2) \succeq \bt \,.
    \end{align}    
\end{thm}
In other words, the GSVs are an extremal case for the diagonal ratios of all possible joint unitary triangularizations.
The joint unitary decomposition that corresponds to these extremal values is the GSVD. 

Following the exposition in \cite{Bai92,PaigeSaunders81}, 
we next review the two forms of the GSVD~--- diagonal and triangular.
The diagonal representation of the GSVD is better known. For a matrix pair $( \bA_1, \bA_2 )$ it is given by \cite{VanLoan76,GolubVanLoan3rdEd}:
\begin{subequations}
\label{eq:GSVD:diagonal}
\noeqref{eq:GSVD:diagonal:1,eq:GSVD:diagonal:2}
\begin{align}
    \bA_1 &= \bU_1 \bL_1 \bX^\dagger 
    ,
\label{eq:GSVD:diagonal:1}
 \\ \bA_2 &= \bU_2 \bL_2 \bX^\dagger
    ,
\label{eq:GSVD:diagonal:2}
\end{align}
\end{subequations}
where $\bU_1$ and $\bU_2$ are unitary, $\bX$ is invertible, and $\bL_1$ and $\bL_2$ are generalized diagonal matrices (\viz, $D_{k;i,j} = 0$ for $i \neq j$, where $D_{k;i,j}$ is the $(i,j)$ entry of $D_k$) with positive diagonal values satisfying: 
\begin{align}
\label{eq:GSVD:normalization}
    \bL_1^\dagger \bL_1 + \bL_2^\dagger \bL_2 = \bI
    ,
\end{align}
the ratios of which are equal to the GSVs: 
\begin{align}
    \frac{D_{1;ii}}{D_{2;ii}} &= \mu_i \left( \bA_1, \bA_2 \right) ,  & i = 1, \ldots, N ,
\end{align}
and are assumed, w.l.o.g., to be ordered non-increasingly.
To obtain the triangular form of the GSVD, apply a QL decomposition\footnote{This decomposition is similar to the QR decomposition, only instead of an upper-triangular matrix, the resulting matrix is lower triangular. This can be achieved, \eg, by applying Gram--Schmidt triangularization to the columns of a matrix, from last to first.} to $\bX$, to attain: 
\begin{subequations}
\label{eq:GSVD:Triangular}
\noeqref{eq:GSVD:Triangular:1,eq:GSVD:Triangular:2}
\begin{align}
    \bA_1 &=
    \bU_1 \bL_1 \bT \bV^\dagger 
    \\ &\triangleq \bU_1 \bT_1 \bV^\dagger
    ,
\label{eq:GSVD:Triangular:1}
\\ 
    \bA_2 &=
    \bU_2 \bL_2 \bT \bV^\dagger 
    \\ &\triangleq \bU_2 \bT_2 \bV^\dagger
    ,
\label{eq:GSVD:Triangular:2}
\end{align}
\end{subequations}
where $\bT$ is upper triangular and $\bV$ is unitary.
By denoting $\bT_1 \triangleq \bL_1 \bT$ and $\bT_2 \triangleq \bL_2 \bT$, we attain the triangular form of the GSVD, 
which is, in turn, a special case of \eqref{eq:STUD}.

%%%%%%%%%%%%%%%%%%%%%%%%%%%%%%%%%%%%%%%%%%%%%%%%%%%%%%%%%%%%%%%%%%%%%%%%%%%%%%%%%%%%%%%%%%%%%%%%%%%%%%%%%%%%%%%%%%%%%%%%%

\section{The MIMO Wiretap Capacity Revisited}
\label{s:CapacityRevisited}

In this section we re-derive the explicit capacity expression of Bustin \etal \cite{Wiretap_BustinEURASIP} 
for the MIMO wiretap channel under a covariance constraint \eqref{eq:covariance_constraint}
in terms of the GSVD.
While we do not establish a new  capacity result, our approach of simultaneous unitary triangularization will lead to a simplified representation of the optimal covariance matrix as well as 
layered coding schemes, as will be discussed in the subsequent sections.

The following augmented matrix structure, which serves as the MIMO channel analogue of 
the minimum mean square error (MMSE) variant of decision feedback equalization for linear time-invariant systems~\cite{CDFE-PartI}, 
will be instrumental throughout this work.

\begin{defn}[Effective MMSE channel matrix]
\label{def:VBLAST_matrix}
    Let $\bH$ be a channel matrix of dimensions $\Nb \times \Na$ 
    and let $\bK$ be the $\Na \times \Na$ input covariance matrix used over this channel.
    Then, the corresponding \emph{effective MMSE channel matrix} is the $(\Na+\Nb) \times \Na$ matrix
      \begin{align}
      \label{eq:G_B}
	  \bG\left( \bH, \bK \right)
          \triangleq
          \begin{pmatrix}
              \bH \bK^{1/2}
           \\ \bI
          \end{pmatrix}
          ,
      \end{align}
      where $\bI$ is the identity matrix of dimension $\Na$ and $\bK^{1/2}$ is any matrix $\bB$ satisfying $\bB \bB^\dagger = \bK$.\footnote{Such a $\bB$ can always be constructed, \eg, using the Cholesky decomposition or unitary diagonalization.}
\end{defn}

This definition naturally lends itself to an MMSE (capacity-achieving) variant of the V-BLAST/GDFE scheme~\cite{HassibiVBLAST}, as will be described in \secref{ss:MIMO_P2P:SIC}. See also~\cite{UCD}, \cite{STUD:SP}, \cite{JET:SeveralUsers2015:FullPaper} for further explanations.

Construct the effective MMSE matrices $\bG_B = \bG(\bH_B, \bK)$ and $\bG_E = \bG(\bH_E, \bK)$,
where
$\bK$ is subject to the constraining matrix $\obK$~\eqref{eq:covariance_constraint}: $\bK \preceq \obK$.

Now, apply some joint unitary triangularization \eqref{eq:GTD}: 
\begin{subequations}
\label{eq:EffMatGSVD}
\noeqref{eq:EffMatGSVD:Bob,eq:EffMatGSVD:Charlie}
\begin{align}
 \bG_B & = \bU_B \bT_B \bV_A^\dagger \,,
\label{eq:EffMatGSVD:Bob}
 \\* 
 \bG_E & = \bU_E \bT_E  \bV_A^\dagger \,,
    \label{eq:EffMatGSVD:Charlie}
\end{align}
\end{subequations}
where $\bU_B$, $\bU_E$ and $\bV_A$ are unitary, and $\bT_B$ and $\bT_E$ are generalized upper triangular (recall \defnref{def:UpperTriangular}).

Let $\{b_i\}$ and $\{e_i\}$ denote the diagonal values of $\bT_B$ and $\bT_E$, respectively, 
where, as explained in \secref{ss:STUD}, these values can be designed by varying $\bV_A$.
Using the fact that the absolute value of a determinant of a unitary matrix is equal to~1, 
and the fact that the determinant of a triangular matrix is equal to the product of its diagonal values, 
the Gaussian MI \eqref{eq:Gaussian_MI} can be expressed as: 
\begin{subequations}
\label{eq:sum_rates:Bob}
\noeqref{eq:sum_rates:Bob:G,eq:sum_rates:Bob:b_i}
\begin{align} 
    I(\bH_B,\bK) &= \log \left| \bG_B^\dagger \bG_B \right| 
\label{eq:sum_rates:Bob:G}
 \\ &= \sum \log b_i^2 ,
\label{eq:sum_rates:Bob:b_i}
\end{align}
\end{subequations}
and similarly for Eve:
\begin{align} 
\label{eq:sum_rates:Eve}
    I(\bH_E,\bK) &= \log \left| \bG_E^\dagger \bG_E \right| 
 \\ &= \sum \log e_i^2 .
\end{align}
Hence, their difference \eqref{eq:Gaussian_MI_difference} is given by 
\begin{align} 
\label{eq:I_S_b_c}
    I_S(\bH_B, \bH_E, \bK) =  \sum_{i = 1}^\Na \log \frac{b_i^2}{e_i^2}  \,.
\end{align}

Note that the expression in \eqref{eq:I_S_b_c} holds for any unitary matrix $\bV_A$ in \eqref{eq:EffMatGSVD}. Indeed, as we shall see later, this flexibility in choosing $\bV_A$ can lead to different design tradeoffs in our layered coding schemes. Nevertheless, to derive an explicit capacity expression we specialize $\bV_A$ to be the right unitary matrix of the GSVD~\eqref{eq:GSVD:Triangular}, until the end of the section. The corresponding GSVs are hence equal to 
\begin{align}
    \mu_i\left( \bH_B, \bH_E, \bK \right) &\triangleq \mu_i \left( \bG_B, \bG_E \right) 
 \\ &= \frac{b_i}{ e_i} \,,
\end{align}
where we use the notation $\mu_i\left( \bH_B, \bH_E, \bK \right)$ to emphasize the dependence in $\bK$.
Without loss of generality, we assume that the GSV vector is non-increasing. 

In terms of the GSVs, we can rewrite \eqref{eq:wiretap_covariance} as:
\begin{align} 
    \label{eq:wiretap_covariance_GSVD}
    C_S(\bH_B,\bH_E,\obK) = \max_{ \bK \preceq \obK} \sum_{i = 1}^\Na \log \mu_i^2 \left( \bH_B, \bH_E, \bK \right)
   . 
\end{align}
Indeed, in these terms the MIMO wiretap capacity can be expressed as follows.
\begin{thm}[MIMO wiretap capacity under a covariance constraint \cite{Wiretap_BustinEURASIP}]
\label{thm:Bustin}
    The secrecy capacity under a covariance matrix constraint  $\obK$ is given by 
    \begin{subequations}
    \label{eq:corol:capacity:GSV:2_expressions}
    \noeqref{eq:corol:capacity:GSV,eq:corol:capacity:GSV:with_Lb}
    \begin{align}
	C_S(\bH_B,\bH_E,\obK)
	&= \sum_{i = 1}^\Na \left[ \log \mu_i^2\left( \bH_B, \bH_E, \obK \right) \right]_+ 
    \label{eq:corol:capacity:GSV}
     \\ &= \sum_{i = 1}^\Lb \log \mu_i^2\left( \bH_B, \bH_E, \obK \right) .
    \label{eq:corol:capacity:GSV:with_Lb}
    \end{align}
    \end{subequations}
\end{thm}
This explicit capacity expression along with the optimal covariance matrix $\bK \preceq \obK$ were established by Bustin \etal~\cite{Wiretap_BustinEURASIP} using the channel enhancement technique along with vector extensions of the mutual information--minimum mean-square error (I--MMSE) relation.
We present an alternative proof of this result using a direct approach: once the optimization problem \eqref{eq:wiretap_covariance} is stated, it can be solved by linear algebra and elementary calculus only. The key to our proof 
is the following lemma.
\begin{lemma} \label{lem:positive}
Let $\obK$ and $\bK$ be two matrices satisfying \mbox{$\bzero \preceq \bK \preceq \obK$}. Then for all $i=1,\ldots,\Na$,
\begin{align} 
    \left| \log \mu_i(\bH_B,\bH_E,\obK) \right| \geq \left| \log \mu_i(\bH_B,\bH_E,\bK) \right|
    .
\end{align}
\end{lemma} 
That is, as we ``decrease'' the input covariance, the GSVs move towards $\mu_i=1$. The proof, which appears in \appref{app:KochmanLemma}, uses standard matrix calculus to show that  
 the differential of the \mbox{$i$-th} GSV, $d \mu_i$, with respect to a change in the covariance matrix $d \bK$, is given by
    \begin{align}
        d \mu_i = \left( \mu_i^2 - 1 \right) \cdot \gamma_i (d\bK)   \,,
    \end{align}
    where $\gamma_i (d\bK) \geq 0$ for $d\bK \succeq \bzero$. Or to put it differently, \mbox{$d \mu_i > 0$} for $\mu_i > 1$, and $d \bmu_i < 0$ for $\mu_i < 1$.

By \lemref{lem:positive}, clearly \thmref{thm:Bustin} gives an upper bound on the capacity. 
To see that it is achievable, consider the matrix:
\begin{align} 
\label{eq:K_B}
    \bK = \obK^{1/2} \bV_A \bI_B \bV_A^\dagger \obK^{1/2 \dagger} ,
\end{align}
where $\bV_A$ is the right unitary matrix of the triangular form of the GSVD \eqref{eq:GSVD:Triangular}, 
$\bI_B$ is a diagonal matrix whose first $\Lb$ diagonal values (corresponding to GSVs that are greater than 1) are equal to 1, 
and the remaining $\Le$~--- to 0.
Trivially, $\bK \preceq \obK$. 
The choice of $\bK$ effectively truncates the GSVs of $\obK$:
 \begin{subequations}
 \begin{align}
      \log  \mu_i^2\left( \bH_B, \bH_C, \bK \right) = 
      \left[ \log \mu_i^2\left( \bH_B, \bH_C, \obK \right)  \right]_+ 
      .
    \label{eq:corol:capacity:GSV:nullified_cov}
    \end{align}
    \end{subequations}
This is formally proved in \appref{app:GSVclipping}.

\begin{remark}
    The optimal covariance matrix $\bK$ \eqref{eq:K_B} is denoted by $\bK_x^*$ in \cite{Wiretap_BustinEURASIP}, where it is given in terms of the diagonal form of the GSVD~\eqref{eq:GSVD:diagonal}:\footnote{In \cite{Wiretap_BustinEURASIP} a specific choice of $\bK^{1/2}$ was used: the matrix $\bB$ that satisfies $\bB \bB = \bK$.}
    \begin{align}
    \label{eq:K_B:Bustin}
        \bK = \obK^{1/2} \bY 
        \begin{bmatrix}
            \left( \bY_B^\dagger \bY_B \right)^{-1} & \bzero_{\Lb \times \Le}
         \\ \bzero_{\Le \times \Lb} & \bzero_{\Le \times \Le}
        \end{bmatrix}
        \bY^\dagger \obK^{\dagger/2}
        \,, 
    \end{align}
    where $\bY = \bX^{-\dagger}$ and $\bX$ is the right invertible matrix of \eqref{eq:GSVD:diagonal}, 
    $\bY_B$ is the sub-matrix composed of the first $\Lb$ columns of $\bY$, and $\bzero_{m \times n}$ denotes the all-zero matrix of dimensions \mbox{$m \times n$}.
Comparing \eqref{eq:K_B} and \eqref{eq:K_B:Bustin}, it is evident that using the triangular form of the GSVD indeed simplifies the representation over using the diagonal one.
\end{remark}

\begin{remark} 
    One may wonder why, of all possible choices of $\bV_A$, the capacity is given in terms of the GSVD. An intuitive reason is as follows. 
    By the majorization condition~\eqref{eq:majorization_gsvd}, the GSV vector is extremal among all possible diagonals. In particular, for any $\bV_A$,
    \begin{align}
        \sum_{i=1}^{\Na} \left[ \log \mu_i^2 \right]_+ \geq \sum_{i=1}^{\Na} \left[ \log \frac{b_i^2}{e_i^2} \right]_+ 
        . 
    \end{align} 
    Thus, the sum \eqref{eq:corol:capacity:GSV} is larger than the sum over diagonal ratios induced by other triangular decompositions.

\end{remark}

\begin{remark} Using \eqref{eq:wiretap_power}, the capacity of the MIMO wiretap channel under a power constraint $P$ can be written as 
    \begin{align}
    \label{eq:K_B:total_power}
        C_S(\bH_B,\bH_C,P) & = \max_{\bK: \trace\{\bK\} = P} \sum_{i = 1}^\Na \left[ \log \mu_i^2\left( \bH_B, \bH_C, \bK \right) \right]_+ .
    \end{align}
\end{remark}

\begin{remark}
\label{rem:GSVD:truncation}
    For the optimal $\bK$~\eqref{eq:K_B}, all the GSVs are greater or equal to 1. To the contrary, assume that some are strictly smaller than 1; then, we can use a matrix $\bK$ with the appropriate directions ``nullified''. Such a ``truncated'' matrix will satisfy the covariance constraint while improving the achievable secrecy rate of the scheme, in contradiction to the assumption.
    \emph{A fortiori}, under a power constraint, the power saved by such a truncation can be allocated to ``useful'' directions.
\end{remark}

%%%%%%%%%%%%%%%%%%%%%%%%%%%%%%%%%%%%%%%%%%%%%%%%%%%%%%%%%%%%%%%%%%%%%%%%%%%%%%%%%%%%%%%%%%%%%%%%%%%%%%%%%%%%%%%%%%%%%%%%%

\section{Scalar Transmission over MIMO Channels}
\label{ss:MIMO_P2P:SIC}

In this section we briefly review the connection between matrix decompositions and scalar transmission schemes, 
without secrecy requirements. For a more thorough account, the reader is referred to \cite{UCD,STUD:SP,JET:SeveralUsers2015:FullPaper}.

In this work we shall assume all the scalar codes to be Gaussian, as defined next.

\begin{defn}[Gaussian codebook]
\label{def:GaussianCodes}
    A Gaussian codebook of length $n$, rate $R$ and power $P - \eps$, where $\eps > 0$, 
    consists of $\left\lceil 2^{nR} \right\rceil$ codewords of length $n$, 
    denoted by $x^n\left( 1 \right), x^n\left( 2 \right), \ldots, x^n\left( \left\lceil 2^{nR} \right\rceil \right)$. The entries of all the codewords, 
    $\{x_t\left( i \right) | t = 1, \ldots, n \,;\, i = 1, \ldots, \left\lceil 2^{nR} \right\rceil \}$, 
    are i.i.d.\ with respect to a Gaussian distribution with zero mean and variance~\mbox{$P - \eps$}.
\end{defn}
\begin{remark}
    In the sequel, with a slight abuse of notation, we shall refer to such codes as Gaussian codes of power $P$ (where $\eps$ will serve as an implicit design parameter).
\end{remark}

Consider the channel \eqref{eq:Bob's_channel}. 
Construct the effective MMSE matrix $\bG_B = \bG(\bH_B, \bK)$ as in \defnref{def:VBLAST_matrix}, 
and choose some unitary matrix~$\bV_A$.

Apply the GTD \eqref{eq:GTD} to $\bG_B$ with $\bV_A$ as the right matrix:
\begin{align}
\label{eq:P2P:GTD}
    \bG_B = \bU_B \bT_B \bV_A^\dagger .
\end{align}
 Now let $\tbx$ be a vector of standard Gaussian variables, and set  
   \begin{align} \label{eq:tx}
        \bx = \bK^{1/2} \bV_A \tbx \,.
    \end{align}
Denote by $\tbU_B$ the $\Nb \times \Na$ upper-left sub-matrix of $\bU_B$, 
and define 
\begin{align}
\label{eq:tbT}
    \tbT_B = \tbU_B^\dagger \bH_B \bK^{1/2} \bV_A. 
\end{align}
The following lemma, whose proof can be found in~\cite{HassibiVBLAST}, \cite[Lemma~III.3]{UCD}, \cite[Appendix~I]{GDFE-HeshamElGamal}, provides the connection between the elements of $\bT_B$ and $\tbT_B$.
\begin{lemma}
\label{lem:bT_tbT_relation}
    Denote by $[\bT_B]$ the $\Na \times \Na$ upper-triangular sub-matrix composed of the first $\Na$ rows of $\bT_B$~\eqref{eq:P2P:GTD}.\footnote{Since $\bT_B$ is full rank, $[\bT_B]$ is full rank too, and hence also invertible. Further, its diagonal elements are greater or equal to 1 due to the block $\bI$ in the construction of $\bG_B$.} Then, $\tbT_B$~\eqref{eq:tbT} is equal to
    \begin{align}
        \tbT_B = [\bT_B] - [\bT_B]^{-\dagger}.
    \end{align}
    In particular, 
    \begin{align}
    \label{eq:bT_tbT_relation}
        \tT_{B;i,j} = 
        \begin{cases}
            T_{B;i,j} & i < j
	 \\ T_{B;i,j} - 1 / T_{B;i,j} & i = j
        \end{cases}
    \end{align}
    where $T_{B;i,j}$ and $\tT_{B;i,j}$ are the $(i, j)$ entries of the matrices $\bT_B$ and $\tbT_B$, respectively.
\end{lemma}

Let
    \begin{subequations}
    \label{eq:rx}
    \noeqref{eq:rx:Uy,eq:rx:Ux+Uz,eq:rx:Tx}
    \begin{align} 
        \tby_B &= \tbU_B^\dagger \by_B 
    \label{eq:rx:Uy}
     \\ &= \tbU_B^\dagger \bH_B \bK^{1/2} \bV_A \tbx + \tbU_B^\dagger \bz_B
    \label{eq:rx:Ux+Uz}
     \\ &= \tbT_B \tbx + \tbz_B \,.\ \ \ 
    \label{eq:rx:Tx}
    \end{align}
    \end{subequations}
  Since $\tbU_B$ is not unitary, the statistics of $\tbz_B \triangleq \tbU_B^\dagger \bz_B$ differ from those of $\bz_B$,
    and its covariance matrix is given by \mbox{$\bK_{\tbz_B} \triangleq \tbU_B \tbU_B^\dagger$}.
Now, for $i=1,\ldots,\Na$, define [recall~\eqref{eq:bT_tbT_relation}]
\begin{subequations}
\label{eq:successive}
\noeqref{eq:successive:y-Tx,eq:successive:explicit,eq:successive:effective}
\begin{align}
        y'_{B;i} & = \ty_{B;i} - \sum_{\ell = i+1}^\Na T_{B;i,\ell} \tx_\ell 
\label{eq:successive:y-Tx}
     \\ &= \tT_{B;i,i} \tx_i + \sum_{\ell=1}^{i-1}  \tT_{B;i,\ell} \tx_\ell + \tz_{B;i} 
\label{eq:successive:explicit}
     \\ &\triangleq \tT_{B;i,i} \tx_i + z^\eff_{B;i} \,,
\label{eq:successive:effective}
\end{align}
\end{subequations}
$\tz_{B;i}$ and $z^\eff_{B;i}$ are the $i$-th entries of the vectors $\tbz_B$ and $\bz_B^\eff$, respectively, 
and \mbox{$z_{B;i}^\eff \triangleq \sum_{\ell=1}^{i-1}  \tT_{B;i,\ell} \tx_\ell + \tz_{B;i}$} is the resulting  total effective noise vector.

In this scalar channel from $\tx_i$ to $y'_{B;i}$, 
resulting after the subtraction of the previously recovered symbols $\{\tx_\ell | \ell > i\}$, 
we view the remaining symbols $\{\tilde x_\ell | \ell < i \}$ as ``interference'', $\tz_{B;i}$~--- as ``noise'', 
and their sum $z^\eff_{B;i}$~--- as ``effective noise''.
The resulting signal-to-interference-and-noise ratio (SINR) is given by:
    \begin{align}
    \label{eq:SINR}
        \SINR_{B;i} 
        &\triangleq \frac{(\tT_{B;i,i})^2}{K_{\bz_B^\eff;i,i}}
     \\ &\triangleq \frac{(\tT_{B;i,i})^2}{K_{\tbz_B;i,i} + \sum\limits_{\ell=1}^{i-1} ( \tT_{B;i,\ell})^2} \,,
    \end{align}
    where $K_{\bz^\eff_B;i,j}$ and $K_{\tbz_B;i,j}$ 
    denote the $(i,j)$ entries of $\bK_{\tbz_B}$ and $\bK_{\bz^\eff_B}$, respectively.
    The following key result achieves the mutual information \cite{HassibiVBLAST}, \cite[Lemma~III.3]{UCD}, \cite[Appendix~I]{GDFE-HeshamElGamal} and is based on \lemref{lem:bT_tbT_relation}.\footnote{Note that, even though $\tbz_B$ has dependent components, the entries of the effective noise $\bz^\eff_B$, are independent.}
    
    \begin{subequations}
    \label{eq:Bob:SINRs_rates}
    \noeqref{eq:Bob:SINRs_rates:MI,eq:Bob:SINRs_rates:SNR_Bi,eq:Bob:SINRs_rates:SNR}
    \begin{align}
          \CMI{\tx_i}{\by_B}{\tx_{i+1}^\Na}
          &= \MI{\tx_i}{y'_{B;i}}
    \label{eq:Bob:SINRs_rates:MI}
      \\ &= \log (1 + \SINR_{B;i})
    \label{eq:Bob:SINRs_rates:SNR_Bi}
     \\  &= \log \left( b_i^2 \right)
          ,
    \label{eq:Bob:SINRs_rates:SNR}
    \end{align}
    \end{subequations}
    where $\left\{ b_i \right\}$ are the diagonal values of $\bT_B$~\eqref{eq:P2P:GTD} [mind the difference from the diagonal values of $\tbT_B$~\eqref{eq:bT_tbT_relation}], 
    which satisfy 
    \begin{align} 
    \label{eq:P2P:Bob:SNRs}
        b_i^2 = 1 + \SINR_{B;i}
    \end{align}
    and 
    \begin{align}
    \label{eq:Bob:Sum_SINR_rates}
        \sum_{i=1}^\Na \log \left( b_i^2 \right)
        &= \sum_{i=1}^\Na \log \left( 1 + \SINR_{B;i} \right)    
     \\* &= I(\bH_B, \bK) ,
    \end{align}
which equals the channel capacity for the optimal $\bK$.

\begin{figure}[t]
     \centering
    \subfloat[Transmitter]{
     \label{fig:VBLAST:Tx}
	 \psfrag{&E1}{Encoder 1}
	 \psfrag{&E2}{Encoder $\Na$}
	 \psfrag{&V}{$\bK_\bx^{1/2} \bV_A$}
	 \psfrag{&tX1}{$\tx_1$}
	 \psfrag{&tX2}{$\tx_\Na$}
	 \psfrag{&X1}{$x_1$}
	 \psfrag{&X2}{$x_\Na$}
	 \epsfig{file = 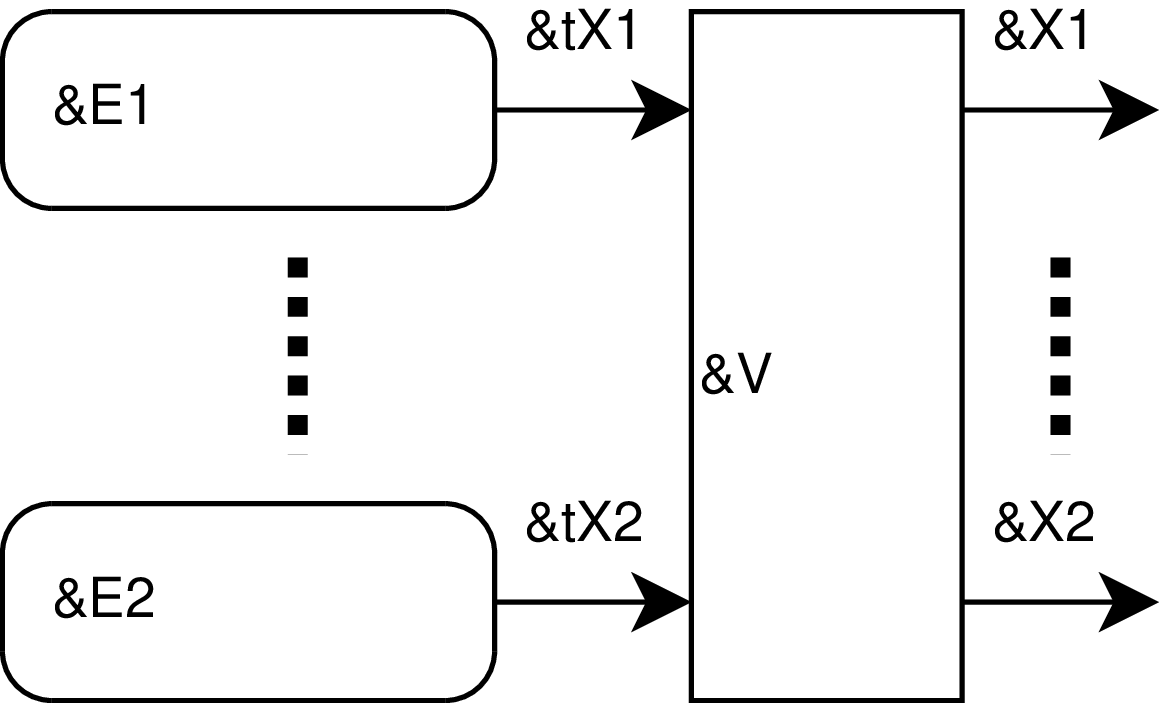, scale = .46}
     }
     \\
  \centering
    \subfloat[Receiver]{
    \label{fig:Rx}
	\psfrag{&x1}{$\thx_\Na$}
	 \psfrag{&x2}{$\thx_{\Na-1}$}
	 \psfrag{&x3}{$\thx_2$}
	 \psfrag{&x4}{$\thx_1$}
	 \psfrag{&y1}{$y_\Nb$}
	 \psfrag{&y2}{$y_{\Nb-1}$}
	 \psfrag{&y3}{$y_1$}
	 \psfrag{&ty1}{$\ty_\Na$}
	 \psfrag{&ty2}{$\ty_{\Na-1}$}
	 \psfrag{&ty3}{$\ty_1$}
	 \psfrag{&y'1}{$y'_\Na$}
	 \psfrag{&y'2}{$y'_\Na$}
	 \psfrag{&y'3}{$y'_1$}
	 \psfrag{&U}{$\tbU_B^\dagger$}
	 \psfrag{&-}{$-$}
	 \psfrag{&D1}{Decoder $\Na$}
	 \psfrag{&D2}{Dec.~$\Na-1$}
	 \psfrag{&D3}{Decoder 1}
	 \psfrag{&G0}{$T_{B;\Na-1,\Na}$}
	 \psfrag{&G1}{$T_{B;1,\Na}$}
	 \psfrag{&G2}{$T_{B;1,\Na-1}$}
	 \psfrag{&G3}{$T_{B;1,2}$}
	 \epsfig{file = 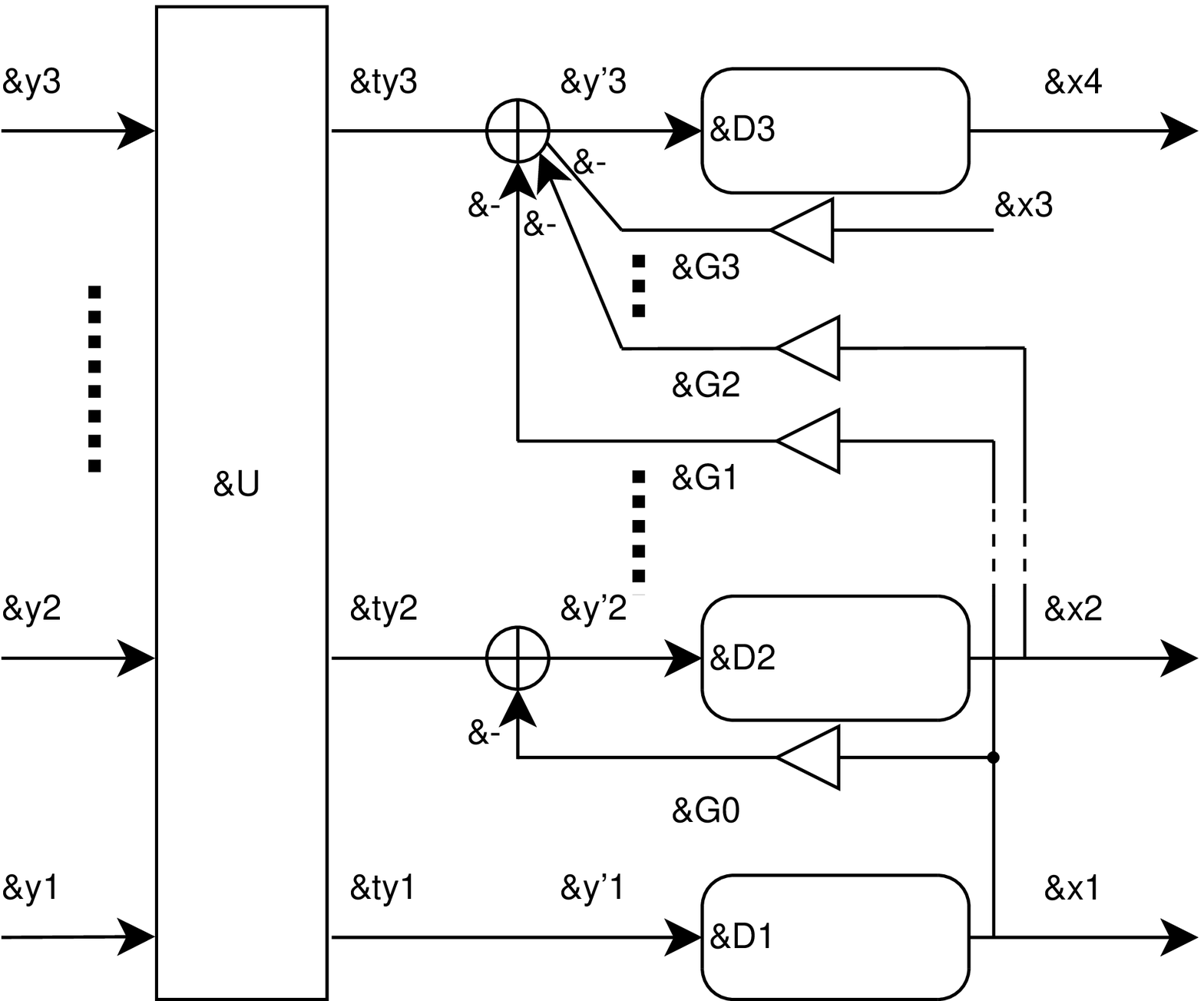, scale = .46}
     }
     \caption{Layered-SIC scheme. $\thx_\ell$ denotes the decoded symbol $\hx_\ell$ at the receiver.}
     \label{fig:VBLAST}
\end{figure}

\begin{figure*}[ht]
     \centering
     \subfloat[Full scheme and channel]{
     \label{fig:SVD@Eve:Full}
	 \psfrag{&ENC1}{\hspace{-4.5mm} Encoder 1}
	 \psfrag{&ENCN}{\hspace{-4.5mm} Encoder $N$}
	 
	 \psfrag{&Va}{$\ubV_A$}
	 \psfrag{&He}{$\bH_E$}
	 \psfrag{&Hb}{$\bH_B$}
	 \psfrag{&Ue}{$\ubU_E^\dagger$}
	 \psfrag{&Ub}{$\tbU_B^\dagger$}
	 
	 \psfrag{&SIC}{SIC}
	 \psfrag{&Rx}{\hspace{-4.5mm} Decoder}
	 
	 \psfrag{&tx1}{$\tx_1$}
	 \psfrag{&txN}{$\tx_N$}
	 \psfrag{&x1}{$x_1$}
	 \psfrag{&xN}{$x_N$}
	 \psfrag{&ze1}[lb]{$z_{E;1}$}
	 \psfrag{&zeN}[lb]{$z_{E;N}$}
	 \psfrag{&ye1}{$y_{E;1}$}
	 \psfrag{&yeN}{$y_{E;N}$}
	 \psfrag{&tye1}{$\ty_{E;1}$}
	 \psfrag{&tyeN}{$\ty_{E;N}$}

	 \psfrag{&zb1}[lb]{$z_{B;1}$}
	 \psfrag{&zbN}[lb]{$z_{B;N}$}
	 \psfrag{&yb1}{$y_{B;1}$}
	 \psfrag{&ybN}{$y_{B;N}$}
	 \psfrag{&tyb1}{$\ty_{B;1}$}
	 \psfrag{&tybN}{$\ty_{B;N}$}

	 \psfrag{&Alice}{Alice}
	 \psfrag{&Bob}{Bob}
	 \psfrag{&Eve}{Eve}
	 
	 \epsfig{file = 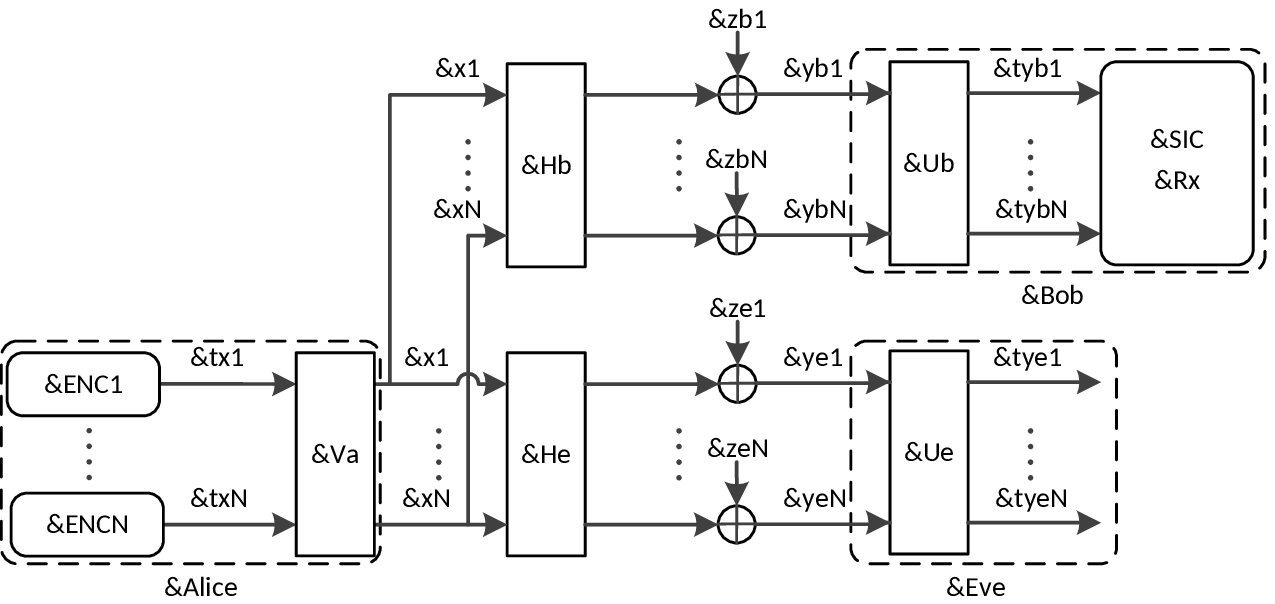, scale = 1.15}
     }
     \\
    \subfloat[Equivalent channel to Eve. $\tz_{E;1}, \ldots, \tz_{E;N}$ are independent 
    unit power AWGNs.]{
     \label{fig:SVD@Eve:EveEquivChannel}
	 \psfrag{&ENC1}{\hspace{-4.5mm} Encoder 1}
	 \psfrag{&ENCN}{\hspace{-4.5mm} Encoder $N$}
	 
	 \psfrag{&tx1}{$\tx_1$}
	 \psfrag{&txN}{$\tx_N$}
	 \psfrag{&tze1}[lb]{$\tz_{E;1}$}
	 \psfrag{&tzeN}[lb]{$\tz_{E;N}$}
	 \psfrag{&ye1}{$\ty_{E;1}$}
	 \psfrag{&yeN}{$\ty_{E;N}$}

	 \psfrag{&d1}{$d_1$}
	 \psfrag{&dN}{$d_N$}
	 
	 \epsfig{file = 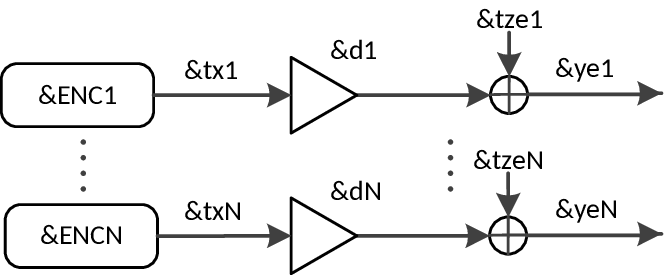, scale = 1.15}
     }
\ \ 
    \subfloat[Equivalent channel to Bob assuming correct past decision at the SIC decoder (which is achieved for a large enough blocklength $n$).]{
     \label{fig:SVD@Eve:BobEquivChannel}
	 \psfrag{&ENC1}{\hspace{-4.5mm} Encoder 1}
	 \psfrag{&ENCN}{\hspace{-4.5mm} Encoder $N$}
	 \psfrag{&DEC1}{\hspace{-4.5mm} Decoder 1}
	 \psfrag{&DECN}{\hspace{-4.5mm} Decoder $N$}
	 
	 \psfrag{&tx1}{$\tx_1$}
	 \psfrag{&txN}{$\tx_N$}
	 \psfrag{&tze1}[lb]{$\tz_{B;1}$}
	 \psfrag{&tzeN}[lb]{$\tz_{B;N}$}
	 \psfrag{&ye1}{$y'_{B;1}$}
	 \psfrag{&yeN}{$y'_{B;N}$}

	 \psfrag{&d1}{\!\!\!\!\!\! $\sqrt{b_1^2-1}$}
	 \psfrag{&dN}{\!\!\!\!\!\! $\sqrt{b_N^2-1}$}
	 
	 \epsfig{file = 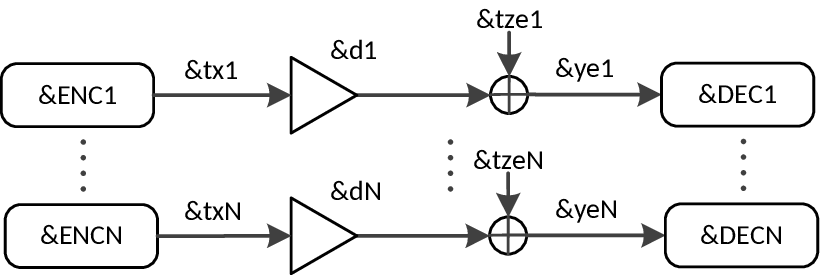, scale = 1.15}
     }
     \caption{Layered-SIC scheme for the MIMO wiretap channel. We assume here $N = \Na = \Nb = \Ne$, 
     for ease of presentation.}
     \label{fig:SVD@Eve}
\end{figure*}

The analysis above immediately gives rise to the following scheme, 
depicted also in \figref{fig:VBLAST}, 
which is, in turn, a variant of the renowned \mbox{V-BLAST/GDFE} scheme \cite{Foschini96,Wolniansky_V-BLAST,CioffiForneyGDFE,HassibiVBLAST}.
\begin{scheme*}[Layered-SIC]
\ 

  \textbf{Offline:}
  \begin{itemize}
  \item
      Select an admissible $\Na \times \Na$ input covariance matrix $\bK$ that satisfies the input constraint.\footnote{More generally, any number $N \geq \rank\{\bK\}$ of scalar codebooks can be used; see \cite{UCD}, \cite{STUD:SP} for details.}
  \item
      Construct the effective MMSE matrix~\eqref{eq:G_B}: $\bG_B = \bG(\bH_B,\bK)$.
  \item
      Select a unitary triangularization~\eqref{eq:GTD} and apply it to the matrix $\bG_B$, 
      as in \eqref{eq:P2P:GTD}, to obtain the unitary matrices $\bU_B$ and $\bV_A$, and 
      the generalized upper-triangular matrix $\bT_B$.
  \item
      Denote the $\Na$ diagonal elements of $\bT_B$ by $\{b_i\}$.
  \item
      Denote by $\tbU_B$ the $\Nb \times \Na$ upper-left sub-matrix of $\bU_B$, 
      and construct the
     corresponding matrix $\tbT_B$ according to \eqref{eq:tbT}: $\tbT_B = \tbU_B^\dagger \bH_B \bK^{1/2} \bV_A$.
  \item
      Construct $\Na$ scalar \emph{Gaussian codes} of length $n$ and unit power
      that are good for SNRs $\{b_i^2 - 1\}$, \ie, codes of rates close to 
      \begin{align}
      \label{eq:SIC_scheme:rates}
          \left\{ R_i \middle| R_i = \log\left( b_i^2 \right) ,\, i \in \{ 1, \ldots, \Na \} \right\} .
      \end{align}
  \end{itemize}

  \textbf{Alice:}
  At each time instant $t = 1, \ldots, n$:
  \begin{itemize}
  \item 
    Forms the vector $\tbx$ of length $\Na$, by taking one sample from each codebook.
  \item
    Attains the vector $\bx$ by multiplying $\tbx$ by $\bV_A$ and $\bK^{1/2}$:
    \begin{align}
    \label{eq:scheme:SIC_scheme:x}
        \bx = \bK^{1/2} \bV_A \tbx .
    \end{align}
  \item
    Transmits $\bx$.
  \end{itemize}

  \textbf{Bob:}
  
    \begin{itemize}
    \item
      At each time instant $t = 1, \ldots, n$, receives $\by_B$ and 
      forms $\tby_B$ according to \eqref{eq:rx}:
        \begin{align}
            \tby_B &= \tbU_B^\dagger \by_B 
         \\ &= \tbT_B \tbx + \tbz_B .
        \end{align}
    \item
      Decodes the $n$-length codewords using SIC, from last (\mbox{$i = \Na$}) to first ($i = 1$):
      Assuming correct decoding of all codebooks $i+1,\ldots,\Na$, Bob forms $y'_{B;i}$ \eqref{eq:successive}:
      \begin{align}
      \label{eq:scheme:SIC:y_B}
        y'_{B;i} = \tT_{B;i,i} \tx_i + z^\eff_i \,,
      \end{align}
      and recovers $\tx_i$.
  \end{itemize} 
\end{scheme*}

By the analysis above, the scheme is optimal in the sense that the sum of codebook rates can approach the channel capacity.

\begin{remark}
    The SIC procedure and the performance analysis of the scheme implicitly assume that the yet-undecoded codebooks 
    can be considered as AWGN, and consequently that each codebook should be capacity achieving for an AWGN channel.
    This is indeed true for Gaussian codes (recall \defnref{def:GaussianCodes}) but not for any single-user scalar capacity-achieving codes as is discussed in \secref{s:discussion}.
\end{remark}

%%%%%%%%%%%%%%%%%%%%%%%%%%%%%%%%%%%%%%%%%%%%%%%%%%%%%%%%%%%%%%%%%%%%%%%%%%%%%%%%%%%%%%%%%%%%%%%%%%%%%%%%%

\section{Multi-Stream Schemes for the \\ MIMO Wiretap Channel}
\label{s:schemes}

Equipped with the results presented in the previous sections, 
we describe how to construct multi-stream schemes that achieve the capacity of the MIMO wiretap channel.

We first describe a scheme in which the channel to Eve is effectively diagonalized, in \secref{s:SVD@Eve}. This particular choice facilitates the proof of both weak and strong secrecy guarantees over this channel. We then extend this result in \secref{s:GeneralScheme}, by proving that any joint triangularization~\eqref{eq:STUD} can be used to construct a multi-stream capacity-achieving scheme.

%------------------------------------------------------------------------------

\subsection{Orthogonalizing Eve's Channel}
\label{s:SVD@Eve}

We now present a simple adaptation of the layered-SIC scheme of \secref{ss:MIMO_P2P:SIC} to the MIMO wiretap setting, depicted also in \figref{fig:SVD@Eve}, 
that achieves the secrecy capacity of the channel using scalar wiretap codes.

To this end, we note that the layered-SIC scheme is capacity-achieving (without secrecy constraints) for any choice of $\bV_A$ in~\eqref{eq:P2P:GTD}.
In particular, we can choose this matrix to be the unitary matrix that diagonalizes Eve's effective channel matrix, namely, the right matrix of the SVD of Eve, denoted by $\ubV_A$:
\begin{align}
\label{eq:SVD@Eve:SVDofHe}
    \bH_E \bK^{1/2} = \ubU_E \ubD_E \ubV_A^\dagger \,.
\end{align}
Applying this $\ubV_A$ to $\bH_E$ (followed by $\bK^{1/2}$) provides effective parallel scalar independent channels to Eve, 
of SNRs 
$\{d_i^2\}$, 
where $\{d_i\}$ are the diagonal values of $\bD_E$, which constitute the singular values of $\bH_E \bK^{1/2}$.

The following simple lemma summarizes the connection between the SVDs of the effective channel matrix~\eqref{eq:SVD@Eve:SVDofHe} and the effective MMSE channel matrix $\bG_E(\bH_E,\bK)$.
\begin{lemma}[Connection to effective MMSE matrix]
\label{rem:SVD:equivalance}
    The SVD of the effective MMSE matrix $\bG_E = \bG(\bH_E, \bK)$ (recall \defnref{def:VBLAST_matrix}) is given by 
    \begin{align}
    \label{eq:bG_E:SVD}
        \bG_E = \bU_E \bD_E \bV_A ,
    \end{align}
    where $\bD_E$ is a generalized diagonal matrix (\viz, $D_{E;i,j} = 0$ for $i \neq j$); denote its diagonal elements by $\{e_i\}$.
    
    The SVD of $\bG_E$~\eqref{eq:bG_E:SVD} 
    is connected to the SVD of $\bH_E \bK^{1/2}$ \eqref{eq:SVD@Eve:SVDofHe} as follows.
    Define $d_i = 0$ for $i > \Na$, 
    and note that $e_i = 1$ for $i > \Na$.
    Define further $\Lambda_E$ as the generalized diagonal matrix of dimensions $\Ne \times \Na$ 
    whose diagonal is equal to $\left( \frac{d_1}{e_1}, \ldots, \frac{d_r}{e_r}  \right)$,
    where $r = \min \{ \Na, \Ne \}$.
    Then, 
    \begin{enumerate}
    \item 
	$\ubV_A = \bV_A$, \ie, $\bG_E$ and $\bH_E \bK^{1/2}$ are diagonalized by the same right matrix.
    \item 
	$1+d_i^2 = e_i^2\,, \qquad i=1,\ldots,N_A$.
    \item 
	$\tbU_E = \ubU_E \Lambda_E$, where $\tbU_E$ is the $\Ne \times \Na$ upper-left sub-matrix of $\bU_E$.
    \end{enumerate}
\end{lemma}

The respective decomposition of $\bG_B$ is as in \eqref{eq:P2P:GTD}, where the diagonal values of the resulting generalized triangular matrix $\bT_B$ are $\{b_i\}$.

Since Eve observes parallel independent channels, using scalar wiretap codes over these channels, that are matched to the SNRs to Eve, 
$\{ d_i^2 \}$, 
guarantees the secrecy of the scheme.
Moreover, by using wiretap codes that work with respect to the SNRs to Bob of \eqref{eq:P2P:Bob:SNRs}, the secrecy capacity is achieved.
This is formally stated in the following theorem.

\begin{thm}
\label{thm:SVD@Eve}
  The layered-SIC scheme of \secref{ss:MIMO_P2P:SIC} achieves the secrecy capacity under a covariance constraint 
    $C_S \left( \bH_B, \bH_E \right. \obK)$ 
    by using:
    \begin{itemize}
    \item 
        The optimal input covariance matrix $\bK$ of \eqref{eq:K_B}. %under a covariance constraint, 
    \item
        Choosing $\bV_A$ of the SVD of $\bH_E \bK^{1/2}$ \eqref{eq:SVD@Eve:SVDofHe}.
    \item
        Scalar Gaussian capacity-achieving wiretap codes 
        that are designed for the Bob--Eve SNR-pairs $\left\{ \left( b_i^2-1, d_i^2 \right) \right\}$.
    \end{itemize}
\end{thm}

\begin{proof}
    The proof easily follows by noting that the resulting channel to Eve is diagonal, \ie, parallel scalar AWGN channels.
    Hence, by using independent (wiretap) Gaussian codes, secrecy is guaranteed over the parallel channels. By combining the result of \secref{ss:MIMO_P2P:SIC} for SIC for MIMO channels without secrecy, correct decoding at Bob's end is guaranteed.
    
    \textbf{Codebook construction:}
    $\Na$ Gaussian codebooks $\{\cC_k | k = 1, \ldots, \Na\}$ of length $n$ are generated independently, 
    as in \defnref{def:GaussianCodes}.
    Codebook $\cC_k$ contains $\left\lceil 2^{n R_k} \right\rceil \times \left\lceil 2^{n \tR_k} \right\rceil$ codewords. Each codeword within $\cC_k$ is assigned a unique index pair $\left( m_k, f_k \right)$, where $m_k \in \left\{ 1, \ldots, \left\lceil 2^{n R_k} \right\rceil \right\}$ and 
    \mbox{$f_k \in \left\{ 1, \ldots, \left\lceil 2^{n \tR_k} \right\rceil \right\}$}. 
    With a slight abuse of notation, we shall refer to such codes as wiretap Gaussian codes of rate-pairs 
    $\left\{ \left( R_k, \tR_k \right) \right\}$.
    
    Let $\eps > 0$. Then the rates are chosen as\footnote{To establish weak secrecy, $\tR_k$ can be relaxed to $\tR_k = \log e_k^2 - \eps$. The choice in \eqref{eq:SVD@Eve:proof:Rates:tRk} allows to establish strong secrecy, as is further explained in the sequel.}
    \begin{subequations}
    \label{eq:SVD@Eve:proof:Rates}
    \noeqref{eq:SVD@Eve:proof:Rates:Rk,eq:SVD@Eve:proof:Rates:tRk}
    \begin{align}
        R_k &= \log \frac{b_k^2}{1 + d_k^2} - 2 \eps = \log \frac{b_k^2}{e_k^2} - 2 \eps ,
    \label{eq:SVD@Eve:proof:Rates:Rk}
     \\ \tR_k &= \log (1 + d_k^2) + \eps = \log e_k^2 + \eps .
    \label{eq:SVD@Eve:proof:Rates:tRk}
    \end{align}
    \end{subequations}
    
    \textbf{Encoding (Alice):}
    Constructs $\Na$ codewords $\{\tx_k \in \cC_k | k = 1, \ldots, \Na \}$ as follows.
    $\tx_k$ is chosen from $\cC_k$ according to the sub-message $m_k$ intended to Bob and 
    a fictitious sub-message $f_k$ which is chosen uniformly at random.
    The transmitted signal at every time instant, $\bx$, is then constructed as in the layered-SIC scheme of \secref{ss:MIMO_P2P:SIC}.
    
    \textbf{Decoding (Bob):}
    Bob performs SIC decoding as in the layered-SIC scheme of \secref{ss:MIMO_P2P:SIC} to recover 
    $\{\left( m_k, f_k \right) \}$, and discards $\{f_k\}$.
    Since $R_k + \tR_k < \log b_k^2$ for every $k$, 
    the decoding error probability of Bob can be made arbitrarily small by taking a large enough $n$.
    
    \textbf{Secrecy analysis (Eve):}
    The resulting channel to Eve~\eqref{eq:SVD@Eve:SVDofHe} 
    (depicted also in \figref{fig:SVD@Eve:EveEquivChannel}) 
    is diagonal:
    \begin{align}
        \tby_E = \ubD_E \tbx + \tbz_E ,
    \end{align}
    where $\tbz_E$ is AWGN with zero mean and identity covariance matrix. That is, 
    the effective channel to Eve comprises independent AWGN channels.
    Over the resulting scalar AWGN channels, wiretap Gaussian codes are known to attain strong secrecy~\cite{GaussianStrongSecrecy}, where $\tR_k$ is chosen to be (slightly) above the \emph{channel resolvability}, \ie, $\tR_k = \log (1 + d_k^2) + \eps$ for $\eps > 0$.
    This is a stronger requirement, as opposed to the choice 
    $\tR_k = \log (1 + d_k^2) - \eps$ for $\eps > 0$, which facilitates an easier proof 
    of weak secrecy guarantees for this channel (see, \eg, \cite[Ch.~22]{ElGamalKimBook}).

    \textbf{Total rate:}
    By using \eqref{eq:I_S_b_c}, \eqref{eq:SVD@Eve:proof:Rates:Rk}, 
    the total rate is equal to 
    \begin{align}
        R &= \sum_{k=1}^\Na R_k 
     \\ &= \sum_{k=1}^\Na \left( \log \frac{b_k^2}{e_k^2} - 2 \eps \right)
     \\ &= I_S\left( \bH_B, \bH_E, \bK \right) - 2 \Na \eps . 
    \end{align}
    By choosing the optimal $\bK$, and taking a large enough $n$, 
    this rate can be made arbitrarily close to the secrecy capacity $C_S$ while guaranteeing both weak and strong secrecy.
\end{proof}

\begin{remark}
    In the proofs to follow, with a slight abuse of notation, we shall state the sizes of the codebook without explicitly using the ceiling operation $\lceil \cdot \rceil$, as its effect becomes negligible for large values of $n$.
\end{remark}

\begin{remark}
\label{rem:IndividualPowerConstraints}
    In the celebrated SVD-based scheme for MIMO channels of Telatar~\cite{Telatar99}, 
    the SVD is applied to the \emph{physical} channel matrix %$\bH$:
        $\bH = \bU \bD \bV_A^\dagger$.
    The transmitted signal is then formed according to \eqref{eq:tx}, where the non-unitary matrix $\bK^{1/2}$ (over the effective diagonal channel $\bD$) is diagonal, with entries set by the water-filling solution.
    Thus, the SVD plays two roles: it serves both for reducing the coding task to that of coding over scalar channels and for constructing the optimal input covariance matrix.
   
    In contrast, in \eqref{eq:SVD@Eve:SVDofHe} the SVD is applied to the \emph{effective} channel matrix $\bH_E \bK^{1/2}$, which already includes the non-unitary ``coloring'' part $\bK^{1/2}$. Thus, it is only used for reducing the coding task. This form is more general, in the sense that it allows for a choice of $\bK$ that is not related to a diagonal decomposition of the channel, \eg, subject to individual power constraints, 
    or where the target expression is different, \eg, an MI difference as in this work.
        Finally, note that the rate of \eqref{eq:I_S_b_c} can be achieved using the proposed scheme, even if $\bK$ is suboptimal (when exact calculation of the optimal $\bK$ is hard).
\end{remark}

%%%%%%%%%%%%%%%%%%%%%%%%%%%%%%%%%%%%%%%%%%%%%%%%%%%%%%%%%%%%%%%%%%%%%%%%%%%%%%%%%%%%%%%%%%%%%%%%%%%%%%%%%%%%%%%%%%%%%%%%%
%%%%%%%%%%%%%%%%%%%%%%%%%%%%%%%%%%%%%%%%%%%%%%%%%%%%%%%%%%%%%%%%%%%%%%%%%%%%%%%%%%%%%%%%%%%%%%%%%%%%%%%%%%%%%%%%%%%%%%%%%
%%%%%%%%%%%%%%%%%%%%%%%%%%%%%%%%%%%%%%%%%%%%%%%%%%%%%%%%%%%%%%%%%%%%%%%%%%%%%%%%%%%%%%%%%%%%%%%%%%%%%%%%%%%%%%%%%%%%%%%%%

\subsection{General Multi-Stream Scheme}
\label{s:GeneralScheme}

We next 
show that, in fact, secrecy capacity can be achieved using the layered-SIC scheme and scalar wiretap codes 
for any choice $\bV_A$, and by this generalizing the result of \secref{s:SVD@Eve} to transmission that is not necessarily orthogonal over Eve's channel.
Specifically, %in \secref{ss:Scheme_proof} 
we show that the secrecy capacity can be achieved using 
any joint triangularization of the effective MMSE channel matrices \eqref{eq:EffMatGSVD} (\emph{any} unitary matrix $\bV_A$ at the encoder). In the general case, 
Eve's resulting matrix is triangular and hence denoted by $\bT_E$, as in \eqref{eq:EffMatGSVD:Charlie}.
The diagonal values of $\bT_E$ are denoted by $\{e_i\}$.
The resulting family of schemes includes two important special cases, discussed in \secref{ss:SpecialCasesOfTriangularizations}, in addition to the one introduced in \secref{s:SVD@Eve}.

\begin{thm}
\label{thm:MIMO_WTC:SIC}
  The layered-SIC scheme of \secref{ss:MIMO_P2P:SIC} achieves the secrecy capacity under a covariance constraint 
    $C_S \left( \bH_B, \bH_E \right. \bK)$ 
    by using:
    \begin{itemize}
    \item 
        The optimal input covariance matrix $\bK$ of \eqref{eq:K_B}. 
    \item
        \emph{Any} joint unitary triangularization \eqref{eq:EffMatGSVD}.
    \item
        Scalar Gaussian capacity-achieving wiretap codes 
        that are designed for the Bob--Eve SNR-pairs $\left\{ \left( b_i^2-1, e_i^2-1 \right) \right\}$, 
        where $\{b_i\}$ and $\{e_i\}$ are defined as in \secref{s:CapacityRevisited}.
    \end{itemize}
\end{thm}

We use the following result, proved in \appref{app:SIC_optimality}, for the proof of this theorem, 
which extends beyond the Gaussian wiretap setting, for both the discrete and the continuous cases.

\begin{prop}
\label{thm:superposition}
    Let $p(y_B|x)$ and $p(y_E|x)$ be the transition distributions for the legitimate user (``Bob'') and the 
    eavesdropper (``Eve''), respectively, of a memoryless wiretap channel, 
    where $x$ is the transmitted signal, and $y_B$ and $y_E$ are the channel outputs to Bob and Eve, respectively.
    Let a superposition coding scheme be defined by codes 
    $\left\{ \tx_i : i = 1, \ldots, \Na \right\}$ 
    and a scalar function $\vphi$ 
    such that 
    \begin{align}
    \label{eq:WTC_superoposition_proof:x=phi_oftilde_x}
    x = \vphi \left( \tx_1, \ldots, \tx_\Na \right) .
    \end{align}
    Then, for $\eps > 0$, however small, and for any joint distribution $p(\rvx_1, \ldots, \rvx_\Na)$,
    there exists a scheme which achieves weak secrecy, with the $k$-th codebook conveying a rate: 
    \begin{align}
    \label{eq:Superposition:Rk}
        R_k &= I(\rvx_k; \rvy_B | \rvx_{k+1}^\Na) - I(\rvx_k; \rvy_E | \rvx_{k+1}^\Na) - \eps .
    \end{align}
\end{prop}

\begin{remark}
    The secrecy-proof of this result uses a ``genie-aided'' argument: 
    in the mutual information of the $k$-th codeword recovered by Eve, we provide all previous codewords 
    $\{ \tbx_\ell | \, \ell= k+1, \ldots, \Na \}$ as ``genie'', even though Eve cannot recover these messages.
    Bob, on the other hand, uses successive decoding to recover the messages. 
    Thus, the allocation of rates $\{ R_k \}$ in \eqref{eq:Superposition:Rk} guarantees that all the messages $(m_1,.. m_{\Na})$
    remain jointly secured from the eavesdropper's channel output sequence.
\end{remark}

\begin{proof}[Proof of \thmref{thm:MIMO_WTC:SIC}]
  We specialize the general superposition coding framework of \propref{thm:superposition} to the linear encoder structure and independent Gaussian distributions of $\left( \rvx_1, \ldots, \rvx_\Na  \right)$.
  Use
  \begin{align}
      \bx &= \vphi \left( \tx_1, \ldots, \tx_\Na \right) 
     \\ &= \bK^{1/2} \bV_A \tbx \,,
  \end{align}
  in \eqref{eq:WTC_superoposition_proof:x=phi_oftilde_x}, 
  where the vector $\tbx$ is composed of one symbol from each codebook: $\tbx = (\tx_1, \ldots, \tx_k)^T$.\footnote{Here, in contrast to \appref{app:SIC_optimality}, boldface letters represent spatial vectors and time indices are suppressed.}
  
  Each codebook is a scalar Gaussian wiretap codebook of average unit power.
  The achievable secrecy rate of codebook $k = 1, \ldots, \Na$ is given by \eqref{eq:Superposition:Rk}:
  \begin{subequations}
  \label{eq:GeneralSchemeProof}
  \noeqref{eq:GeneralSchemeProof:a,eq:GeneralSchemeProof:c,eq:GeneralSchemeProof:Rk_diag_values_ratio_explicit}
  \begin{align}
      R_k &= \CMI{\rvx_k}{\rvy_B}{\rvx_{k+1}^\Na} - \CMI{\rvx_k}{\rvy_E}{\rvx_{k+1}^\Na} - \eps
  \label{eq:GeneralSchemeProof:a}
   \\     &= \MI{\rvx_k}{\rvy_{B;k}'} - \MI{\rvx_k}{\rvy_{E;k}'} - \eps
  \label{eq:GeneralSchemeProof:c}
   \\     &= \log \left( b_k^2 \right) - \log \left( e_k^2 \right) - \eps
  \label{eq:GeneralSchemeProof:Rk_diag_values_ratio}
    \\    &= \log \frac{b_k^2}{e_k^2} - \eps
      \,,
  \label{eq:GeneralSchemeProof:Rk_diag_values_ratio_explicit}
  \end{align}
  \end{subequations}
  where \eqref{eq:GeneralSchemeProof:Rk_diag_values_ratio} and \eqref{eq:GeneralSchemeProof:c} are due to \eqref{eq:Bob:SINRs_rates:MI} and \eqref{eq:Bob:SINRs_rates:SNR}, respectively.
  Thus, using the result of \eqref{eq:I_S_b_c}, we can achieve 
  \begin{align} 
    R &= \sum_{k=1}^N R_k 
   \\ &= \sum_{k=1}^N \left[ \log\frac{b_k^2}{e_k^2} \right]_+ - \eps
   \\ &= I_S\left( \bH_B, \bH_E, \bK \right), 
  \end{align}
  and for the optimal covariance matrix $\bK$ the scheme approaches the secrecy capacity.
\end{proof}

\subsection{Important Special Cases}
\label{ss:SpecialCasesOfTriangularizations}

We now present ``special'' choices of $\bV_A$ which provide various advantages.

\subsubsection{Orthogonalizing Eve's channel}
    The scheme of \secref{s:SVD@Eve} is a special case of proposed scheme in this subsection, since, 
    as explained in \lemref{rem:SVD:equivalance}, the unitary matrix $\bV_A$ of the SVD 
    of $\bH_E \bK^{1/2}$ is identical to that of the SVD of $\bG_E$ \eqref{eq:EffMatGSVD:Charlie}.

\subsubsection{Orthogonalizing Bob's channel~--- Avoiding SIC}
Performing SIC adds complexity to the decoder, as well as introduces potential error propagation. We can avoid this by performing SVD with respect to Bob's channel, as opposed to Eve's channel, as done in \secref{s:SVD@Eve}. That is, choose $\bV_A$ such that 
\begin{align}
  \bG_B = \bU_B \bD_B \bV_A^\dagger ,
\end{align}
where $\bD_B$ is diagonal.
As happens with Eve in \secref{s:SVD@Eve}, 
Bob obtains a diagonal equivalent channel, 
where each sub-stream can be decoded independently.

\subsubsection{Avoiding individual bit-loading} 
When using (non-secret) communication schemes based on SVD or QR, as in the layered-SIC scheme, 
the effective sub-channel gains $\{b_i\}$ are different in general.
This requires, in turn, a bit-loading mechanism and the design of codes of different rates 
matching these gains.
By using the GMD, described in \secref{ss:GTD}, instead, a constant diagonal is achieved, 
which translates into equal SNRs for all parallel channels.
This suggests, in turn, that bit-loading can be avoided altogether and that 
the codewords sent over the resulting sub-channels can be drawn from the same codebook.

A similar result can be achieved for the wiretap setting. 
To this end we require the usage of a modular scheme that 
transforms good AWGN codes of a rate close to $\log(b^2)$ for Bob 
into wiretap codes of rates close to $\{\log(b^2/e_i^2)\}$.
This way, after applying the GMD to $\bG_B$, the same AWGN codebook can be used over all sub-channels, where for each sub-channel a different transformation into a wiretap code is used, that depends on its effective SNR to Eve $(e_i^2-1)$.
Indeed, such a modular approach exists; see \secref{s:discussion}.

\begin{remark}
    It is possible to use the same wiretap code without assuming the modular wiretap code construction, by using a joint matrix decomposition that achieves constant diagonals for both triangular matrices simultaneously.
    A construction that essentially achieves this property was proposed in \cite{JET:SeveralUsers2015:FullPaper}.
\end{remark}

%%%%%%%%%%%%%%%%%%%%%%%%%%%%%%%%%%%%%%%%%%%%%%%%%%%%%%%%%%%%%%%%%%%%%%%%%%%%%%%%%%%%%%%%%%%%%%%%%%%%%%%%%%%%%%%%%%%%%%%%%%%%%%%%%%%%%%%%%%%%%%%%%%%%%%%%%%%%%%%%%%%%%%%%%%%%%%

\section{Dirty-Paper Coding Based Schemes}
\label{s:DPC}

In this section we construct the DPC counterparts of the layered-SIC scheme for Gaussian MIMO channels with and without secrecy constraints.
In these variants the successive decoding process of the scalar codes is replaced with a successive encoding one; consequently, all (scalar) codebooks can be recovered in parallel and independently of each other.
The latter makes these variants useful for more complex settings, such as the confidential MIMO broadcast setting treated in \secref{s:BC}.
We start by presenting the DPC-based schemes without secrecy constraints, in \secref{ss:DPC:P2P}.
We then construct a variant for the MIMO wiretap setting, in \secref{ss:DPC:WTC}, which again achieves the secrecy capacity of the channel.

%--------------------------------------------------------------------------------------------------------------

% \subsection{MIMO Point-to-Point}
\subsection{Without Secrecy Constraints}
\label{ss:DPC:P2P}

We now briefly review the DPC variant of the layered-SIC scheme, which is based in turn on \cite{GinisCioffiAsilomar,CaireShamai03} (see also \cite{UCD}).

% \begin{scheme}[MIMO point-to-point via DPC]
% \label{scheme:p2p_dpc}
\break
\begin{scheme*}[Layered-DPC]
\ 

  \textbf{Offline:}
  \begin{itemize}
  \item
      Select an admissible $\Na \times \Na$ input covariance matrix $\bK$ that satisfies the input constraint.
  \item
      Construct the effective MMSE matrix~\eqref{eq:G_B}: $\bG_B = \bG(\bH_B,\bK)$.
  \item
      Select a unitary triangularization~\eqref{eq:GTD} and apply it to the matrix $\bG_B$, 
      as in \eqref{eq:P2P:GTD}, to obtain the unitary matrices $\bU_B$ and $\bV_A$, and 
      the generalized upper-triangular matrix $\bT_B$.
  \item
      Denote the $\Na$ diagonal elements of $\bT_B$ by $\{b_i\}$.
  \item
      Denote by $\tbU_B$ the $\Nb \times \Na$ upper-left sub-matrix of $\bU_B$, 
      and
      construct the
     corresponding matrix $\tbT_B$ according to \eqref{eq:tbT}: $\tbT_B = \tbU_B^\dagger \bH_B \bK^{1/2} \bV_A$.
  \item
      Construct $\Na$ scalar \emph{dirty-paper codes}~\cite{Costa83} of length $n$~--- codes generated via
      random binning with respect to i.i.d.\ Gaussian distributions. 
    Codebook $i$ ($1 \leq i \leq \Na$) is constructed for a channel with AWGN of unit power, 
    SNR $(b_i^2-1)$, interference [recall \eqref{eq:bT_tbT_relation}]
    \begin{align}
    \label{eq:MIMOP2P:SI}
        \sum_{\ell = i+1}^\Na T_{B;i,\ell} \tx_\ell
    \end{align}
    which is available as side information at the transmitter, and rate $R_i$ close to $\log(b_i^2)$ [recall \eqref{eq:SIC_scheme:rates}].
  \end{itemize}

  \textbf{Alice:}
  At each time instant $t = 1, \ldots, n$:
  \begin{itemize}
  \item
    Generates $\tx_i$ from last ($i=\Na$) to first ($i=1$), where $\tx_i$ is generated according to the message to be conveyed and the interference \eqref{eq:MIMOP2P:SI}.
  \item 
    Forms $\tbx$ with entries $\{\tx_i\}$.
  \item
    Attains the vector $\bx$ by multiplying $\tbx$ by $\bV_A$ and $\bK^{1/2}$ as in
    \eqref{eq:scheme:SIC_scheme:x}.
  \item
    Transmits $\bx$.
  \end{itemize}

  \textbf{Bob:}
  
    \begin{itemize}
    \item
      At each time instant $t = 1, \ldots, n$, receives $\by_B$ and 
      forms $\tby_B$ according to \eqref{eq:rx}:
        \begin{align}
            \tby_B &= \tbU_B^\dagger \by_B 
         \\ &= \tbT_B \tbx + \tbz_B .
        \end{align}
    \item
      Decodes the codebooks using dirty-paper decoders, where $\tx_i$ is decoded from $\ty_{B;i}$.
  \end{itemize} 
\end{scheme*}
\vspace{.5\baselineskip}

By using good dirty-paper codes, capacity is achieved; see, \eg, \cite{UCD}.

We further note that codeword $\tx_i$ is recovered from $\ty_{B;i}$ regardless of whether the other codewords $\{\tx_j| j \neq i\}$ were recovered or not.

%--------------------------------------------------------------------------------------------------------------

\subsection{MIMO Wiretap Channel}
\label{ss:DPC:WTC}

By replacing the dirty-paper scalar codes in the layered-DPC scheme of \ref{ss:DPC:P2P}
with scalar dirty-paper wiretap codes \cite{WTC_DPC,WTC_GP}, 
a scheme that approaches the MIMO wiretap secrecy capacity can be constructed.

\begin{thm}
\label{thm:MIMO_WTC:DPC}
  The layered-DPC scheme of \secref{ss:DPC:P2P} achieves the secrecy capacity under a covariance constraint 
    $C_S \left( \bH_B, \bH_E \right. \obK)$ 
    by using:
    \begin{itemize}
    \item 
        The optimal input covariance matrix $\bK$ of \eqref{eq:K_B}. 
    \item
        \emph{Any} joint unitary triangularization \eqref{eq:EffMatGSVD}.
    \item
        Scalar Gaussian dirty-paper wiretap codes, where the \mbox{$i$-th} codebook ($i = 1, \ldots, \Na$) is designed for
        \begin{itemize}\addtolength\itemsep{.4\baselineskip}
        \item
            Bob's SNR of $(b_i^2-1)$ and interference signal $\sum_{\ell = i+1}^\Na T_{B;i,\ell} \tx_\ell$.
        \item
            Eve's SNR of $(e_i^2-1)$. 
          \item
            Rate close to $R_i = \log ( b_i^2 / e_i^2 )$.
        \end{itemize}
    \end{itemize}
\end{thm}

We next prove the existence of such codes and consequently also the result of \thmref{thm:MIMO_WTC:DPC}.

\begin{proof}
    The proof follows by a standard extension of the proof of \thmref{thm:MIMO_WTC:SIC} 
    to the dirty-paper case~\cite{Costa83,WTC_DPC,WTC_GP}.
    
    \textbf{Codebook construction:}
    For each $k=1, \dots, \Na$, we generate a codebook $\cC_k$ of $2^{n (R_k + \tR_k)}$ sub-codebooks, where $n$ is length of the codewords.
    Each such sub-codebook is assigned a unique index pair $(\rvm_k, \rvf_k)$, 
    where
    $\rvm_k \in \{ 1,2,\ldots, 2^{n R_k} \}$ and 
    $\rvf_k \in \{ 1,2,\ldots, 2^{n \tR_k} \}$, 
    and contains $2^{n [ R^U_k - (R_k + \tR_k) ]}$ codewords.
    Each codeword within codebook $k$ is generated independently in an i.i.d.\ manner with respect to a Gaussian distribution $p(\rvu_k)$ with parameters dictated by     
    \begin{subequations}
    \label{eq:GP:R_U}
    \noeqref{eq:GP:R_U:R_U_def,eq:GP:R_U:alpha_def}
    \begin{align}
        \rvu_k &= \tT_{B;k,k} \rvx_k + \alpha_k \sum_{\ell = k+1}^\Na \tilde T_{B;k,\ell} \rvx_\ell \,,
    \label{eq:GP:R_U:U_def}
     \\ \alpha_k &\triangleq \frac{b_k^2 - 1}{b_k^2} \:\:,
    \label{eq:GP:R_U:alpha_def}
    \end{align}
    \end{subequations}
    for zero mean unit power i.i.d.\ Gaussian random variables $\{ \rvx_k | k = 1, \ldots, \Na \}$.

    Note that since in this case the interference (available as side information to Alice) in sub-channel $k$ 
    is composed of messages $\{ x_\ell | \ell = 1, \ldots, \Na \}$, the information carried by the sets 
    $\{ \rvx_\ell | \ell = 1, \ldots, \Na \}$ and $\{ \rvu_\ell | \ell = 1, \ldots, \Na \}$ is the same.

    Let $\eps > 0$. Then the rates are chosen as 
    \begin{subequations}
    \label{eq:GP:Rates}
    \noeqref{eq:GP:Rates:Rk:def,eq:GP:Rates:Rk,eq:GP:Rates:tRk,eq:GP:R_U:R_U_def}
    \begin{align}
        R_k 
        &\triangleq \MI{\rvu_k}{\rvby_B} - \MI{\rvu_k}{ \rvby_E, \rvu_{k+1}^\Na } - \eps
     \\ &= \left[ \MI{\rvu_k}{\rvby_B} - \MI{\rvu_k}{ \rvu_{k+1}^\Na }  \right]
        - \CMI{\rvu_k}{\rvby_E}{\rvu_{k+1}^\Na} - \eps
    \nonumber
     \\ &= \CMI{\rvx_k}{\rvby_B}{\rvx_{k+1}^\Na} - \CMI{\rvx_k}{\rvby_E}{\rvx_{k+1}^\Na} - \eps
     \\ &= \log \frac{b_k^2}{e_k^2} - \eps, 
    \label{eq:GP:Rates:Rk}
     \\ \tR_k &\triangleq \CMI{\rvu_k}{\rvby_E}{\rvu_{k+1}^\Na} - \eps 
        = \CMI{\rvx_k}{\rvby_E}{\rvx_{k+1}^\Na} - \eps
     \\ &= \log e_k^2 - \eps ,
    \label{eq:GP:Rates:tRk}
     \\ R^U_k &\triangleq \MI{\rvu_k}{\rvby_B} - \eps 
     \\ &= \log \left( b_k^2 + \sum_{\ell = k+1}^\Na \left| T_{B;k,\ell} \right|^2 \right) - \eps \,.
    \label{eq:GP:R_U:R_U_def}
    \end{align}
    \end{subequations}

    \textbf{Encoding (Alice):}
    Encoding is carried in a successive manner, from last ($k=\Na$) to first ($k=1$).
    Within codebook $k$, the index of the sub-codebook to be used is determined by the secret message $\rvm_k$ and a fictitious message $\rvf_k$ drawn uniformly over their respective ranges. The codeword $\bu_k$, within sub-codebook $(\rvm_k,\rvf_k)$ that is selected, 
    is the one that is jointly typical with the side information $\sum_{\ell = k+1}^\Na \tilde T_{B;k,\ell} \tx_\ell$.
    If no such codeword $\bu_k$ exists, then the first codeword is selected.

    \textbf{Decoding (Bob):}
    Bob recovers $(\rvm_k,\rvf_k)$ using standard dirty-paper decoding as in \secref{ss:DPC:P2P}, and discards $\rvf_k$.
    The error probability can be made arbitrarily small by taking a large enough $n$.

    \textbf{Secrecy analysis (Eve):}
    As in the proof of \propref{thm:superposition}, we provide $\{ \rvu_\ell | \ell = k+1, \ldots, \Na \}$ as a genie for the secrecy analysis of $\rvu_k$.
    By recalling that $\{ \rvx_\ell | \ell = k+1, \ldots, \Na \}$ and $\{ \rvu_\ell | \ell = k+1, \ldots, \Na \}$ carry the same information, and the linear 
    relation in the definition of $\rvu_k$ \eqref{eq:GP:R_U:U_def}, the secrecy analysis reduces to the analysis in the proof of 
    \propref{thm:superposition}, as appears in \appref{app:SIC_optimality}, specialized to the Gaussian case.
\end{proof}

%%%%%%%%%%%%%%%%%%%%%%%%%%%%%%%%%%%%%%%%%%%%%%%%%%%%%%%%%%%%%%%%%%%%%%%%%%%%%%%%%%%%%%%%%%%%%%%%%%%%%%%%%%%%%%%%%%%%%%%%%

\section{Confidential Broadcast as a Consequence}
\label{s:BC}

In this section we consider the two-user MIMO confidential broadcast scenario.
Namely, ``Eve'' is replaced with ``Charlie'' in \eqref{eq:Eve's_channel}, and the corresponding noise, output and channel matrix are denoted by 
$\bz_C$, $\by_C$ and $\bH_C$, respectively.

We next show that, under the covariance matrix constraint, the rectangular capacity region \eqref{eq:BC_C}, that was established in \cite{ConfidentialMIMO_BC}, 
can be attained as a natural extension of the capacity derivation for the MIMO wiretap channel and the layered DPC scheme proposed in Sections \ref{s:CapacityRevisited} and \ref{s:DPC}, respectively.

%------------------------------------------------------------------------------------------------------------------------------

\subsection{Capacity Region}

We saw in \secref{s:CapacityRevisited} that in order to achieve the secrecy capacity where Charlie takes the role of Eve, 
the GSVD needs to be applied to $(\bG_B, \bG_C)$ and only the sub-channels corresponding to GSVs that are greater than 1 (corresponding to sub-channels with greater SNR to Bob than to Charlie) need to be used, and the rest~--- nullified.

However, we note that, if we were interested in confidential communication with Charlie rather than with Bob, we would get the same solution with the roles of $\bH_B$ and $\bH_C$ reversed. This, in turn, means inversion of the GSVs:
\begin{align} 
    \log \mu_i(\bH_C,\bH_B,\obK) = -  \log \mu_i(\bH_B,\bH_C,\obK). 
\end{align}
In these terms, we can write the rectangular capacity-region of the confidential broadcast channel \eqref{eq:BC_C}, established first in \cite{ConfidentialMIMO_BC}, 
as follows.

\begin{thm} 
\label{thm:BC}
    The capacity region of the confidential MIMO broadcast channel under an input covariance constraint $\obK$ is given by all rates $(R_B,R_C)$ satisfying:
    \begin{subequations}
    \label{eq:BC:capacity}
    \noeqref{eq:BC:capacity:Bob,eq:BC:capacity:Charlie}
    \begin{align}
        R_B & \leq \sum_{i = 1}^\Na \left[ \log \mu_i^2\left( \bH_B, \bH_C, \obK \right) \right]_+ ,
    \label{eq:BC:capacity:Bob}
     \\ R_C & \leq \sum_{i = 1}^\Na \left[ - \log \mu_i^2\left( \bH_B, \bH_C, \obK \right) \right]_+ .
    \label{eq:BC:capacity:Charlie}
    \end{align}
    \end{subequations}
\end{thm}

\begin{remark}
Similarly to the MIMO wiretap channel, the capacity region under a power constraint $P$ is just the union of all (rectangular) regions under a covariance constraint with small enough trace.
\end{remark}

The converse part of this result is trivial by \thmref{thm:Bustin}, since both users attain their individual secrecy capacities. For the direct part, it is tempting to think that since different GSVs are nullified for Bob and for Charlie, Alice can achieve their optimal rates simultaneously by communicating over orthogonal  ``subspaces''. However, since the matrices $\bT_B$ and $\bT_C$ are not diagonal, these ``subspaces'' are not orthogonal, and some more care is needed. 

To this end, in the next section we put into force the layered-DPC scheme of \secref{s:DPC}, 
which allows to recover the sub-message transmitted over each sub-channel independently, without the recovery of other sub-messages (in contrast to the layered-SIC scheme).
This property is required by at least one of the users~---
Bob or Charlie~--- as each of them recovers only a subset of all the transmitted sub-messages.
The derivation of the scheme thus provides a constructive proof for the direct part of \thmref{thm:BC}, which is an alternative to the proof in \cite{ConfidentialMIMO_BC}.

%------------------------------------------------------------------------------------------------------------------------------

\subsection{Capacity Achieving Schemes}

In view of \thmref{thm:STUD} and 
the schemes developed for the MIMO wiretap channel, the result of \secref{s:CapacityRevisited} has a rather intuitive interpretation:
$\bV_A$ of the GSVD is the precoding matrix that designs the ratios between $\{b_i\}$ and $\{c_i\}$ to be 
as large as possible ($\{c_i\}$ replacing $\{e_i\}$), which corresponds to maximizing the achievable secrecy rate to Bob.
In order to achieve Bob's secrecy capacity, only the sub-channels for which the secrecy rate is positive ($b_i > c_i$) need to be utilized.
Allocating the remaining sub-channels to Charlie, on the other hand, attains Charlie's optimal covariance matrix.

Combining the two gives rise to the following scheme, which is a straightforward adaptation of the layered-DPC scheme of \secref{s:DPC} for the wiretap channel.

\begin{scheme*}[Confidential broadcast via layered-DPC]
\label{scheme:bc_dpc}
\ 

  \textbf{Offline:}
  \begin{itemize}
  \item
      Construct the effective MMSE matrix~\eqref{eq:G_B}: $\obG_B \triangleq \bG(\bH_B, \obK)$ and $\obG_C
      \triangleq \bG(\bH_C, \obK)$, where $\obK$ is the constraining matrix.
  \item
      Apply the triangular form of the GSVD~\eqref{eq:GSVD:Triangular} to $(\obG_B, \obG_C)$ as in
      \eqref{eq:P2P:GTD}, to obtain the unitary matrices $\bU_B$, $\bU_C$ and $\bV_A$, and 
      the generalized upper-triangular matrices $\bT_B$ and $\bT_C$.
  \item
      Denote the diagonal elements of $\bT_B$ and of $\bT_C$ by $\{b_i\}$ and $\{c_i\}$, respectively.
  \item
    Denote further the (first) number of indices 
    for which $b_i > c_i$ 
    by $\Lb$.
    The remaining $\Lc = \Na - \Lb$ indices satisfy $c_i \geq b_i$.
  \item 
    Denote by $\tbU_B$ the upper-left $\Nb \times \Lb$ sub-matrix of $\bU_B$, 
    and by $\tbU_C$~--- the upper-right $\Nc \times \Lc $ sub-matrix  of $\bU_C$.
  \item
     Construct $\tbT_B$ and $\tbT_C$ as in \eqref{eq:tbT}: 
     \begin{align}
	 \tbT_B &= \tbU_B^\dagger \bH_B \bK^{1/2} \bV_A \,,
      \\ \tbT_C &= \tbU_C^\dagger \bH_C \bK^{1/2} \bV_A \,.
     \end{align}
  \item
    Construct $\Na$ good scalar dirty-paper wiretap codes of unit power and length $n$, 
    denoted by $\{ \tx_i | i= 1, \ldots, \Na  \}$ (with the time index omitted to simplify notation), 
    generated via random binning
    with respect to i.i.d.\ Gaussian distributions, 
    as follows. 
    \begin{itemize}
    \item
        The first $\Lb$ codes are intended for Bob:
        Codebook $\tx_i$ ($1 \leq i \leq \Lb$) of a rate close to $R_i = \log\left( b_i^2 / c_i^2 \right)$
        is constructed for an AWGN channel to Bob of SNR $b_i^2-1$, 
        and interference:
        \begin{align}
            \sum_{\ell = i+1}^\Na T_{B;i,\ell} \tx_\ell \,, 
        \end{align}
        and for an AWGN channel to Charlie of SNR $c_i^2-1$.
    \item
        The remaining $\Lc$ codes are intended for Charlie:
        Codebook $\tx_i$ ($\Lb+1 \leq i \leq \Na$) of a rate close to $R_i = \log\left( c_i^2 / b_i^2 \right)$
        is constructed for an AWGN channel to Charlie of SNR $c_i^2-1$ 
        and interference:
        \begin{align}
            \sum_{\ell = i+1}^\Na T_{C;i,\ell} \tx_\ell \,, 
        \end{align}
        and for an AWGN channel to Bob of SNR $b_i^2-1$. 
    \end{itemize}
  \end{itemize}

  \textbf{Alice:}
  At each time instant $t = 1, \ldots, n$:
  \begin{itemize}
  \item
    Generates $\tx_i$ from last ($i=\Na$) to first ($i=1$), where $\tx_i$ is generated according to the message to be conveyed and the signals $\{\tx_\ell | \ell = i+1,\ldots,\Na\}$. 
  \item 
    Forms $\tbx$ with entries $\{\tx_i\}$.
  \item
    Attains the vector $\bx$ by multiplying $\tbx$ by $\bV_A$ and $\bK^{1/2}$ as in
    \eqref{eq:scheme:SIC_scheme:x}.
  \item
    Transmits $\bx$.  
  \end{itemize}

  \textbf{Bob:}
    \begin{itemize}
    \item
      At each time instant $t = 1, \ldots, n$, receives $\by_B$ and 
      forms $\tby_B$ according to \eqref{eq:rx}:
        \begin{align}
            \tby_B &= \tbU_B^\dagger \by_B 
         \\ &= \tbT_B \tbx + \tbz_B .
        \end{align}

    \item
      Decodes codebooks $i =  1, \ldots, \Lb$ using dirty-paper decoders, where $\tx_i$ is decoded from $\ty_{B;i}$.
  \end{itemize} 

  \textbf{Charlie:}
    \begin{itemize}
    \item
      At each time instant forms 
      \begin{align}
          \tby_C &= \tbU_C^\dagger \by_C 
       \\ &= \tbT_C \tbx + \tbz_C .
      \end{align}

    \item
      Decodes codebooks $i = \Lb + 1, \ldots, \Na$ using dirty-paper decoders, where $\tx_i$ is decoded from $\ty_{C;(i - \Lb)}$.
  \end{itemize} 
\end{scheme*}

The following theorem proves that this scheme allows both users to attain their respective secrecy capacities \emph{simultaneously}, providing a proof for \thmref{thm:BC}.

\begin{thm}
\label{thm:BC:DPC}
    The layered-DPC confidential broadcast scheme achieves the secrecy capacity region under a covariance constraint \eqref{eq:BC:capacity} 
    by:
    \begin{itemize}
    \item
        Using scalar Gaussian dirty-paper wiretap codes intended for Bob, as follows, where the $i$-th codebook ($i = 1, \ldots, \Lb$) is designed for:
        \begin{itemize}\addtolength\itemsep{.4\baselineskip}
        \item
            Bob's SNR of $(b_i^2-1)$ and interference signal $\sum_{\ell = i+1}^\Na T_{B;i,\ell} \tx_\ell$.
        \item
            Charlie's SNR of $(c_i^2-1)$. 
          \item
            Rate close to $R_i = \log ( b_i^2 / c_i^2 )$.
        \end{itemize}
    \item
        Using scalar Gaussian DPC wiretap codes intended for Charlie, as follows, where the $i$-th codebook ($i = \Lb+1, \ldots, \Na$) is designed for:
        \begin{itemize}\addtolength\itemsep{.4\baselineskip}
        \item
            Charlie's SNR of $(c_i^2-1)$ and interference $\sum_{\ell = i+1}^\Na T_{C;i,\ell} \tx_\ell$.
        \item
            Bob's SNR of $(b_i^2-1)$. 
          \item
            Rate close to $R_i = \log ( c_i^2 / b_i^2 )$.
        \end{itemize}
    \end{itemize}
\end{thm}

\begin{proof}[Proof sketch]
        We start by noting that since the capacity region is rectangular, it suffices to show how to approach the corner point of this region.
        The proof relies on the fact that in the layered-DPC scheme for the MIMO wiretap channel of \secref{s:DPC}, each sub-codebook is recovered independently, regardless of the other sub-codebooks.
        Hence, the proof of the decodability and secrecy analysis for Charlie are the same as in the proof of \thmref{thm:MIMO_WTC:DPC} (with Charlie being the ``legitimate'' user).
        In the treatment for Bob, a small variation is needed: the interference over sub-channel $i$ ($1 \leq i \leq \Lb$) is composed of both, messages intended for Charlie, $\tx_{\Lb+1}^\Na$, and messages intended for Bob, $\tx_{i+1}^\Lb$. Thus, the DPC for Bob is carried with respect to both of these interferences, 
        and the decodability and secrecy analysis follow as in the proof of \thmref{thm:MIMO_WTC:DPC}. 
\end{proof}

 \begin{remark}[Replacing DPC with SIC]
    DPC was used in the layered-DPC scheme for both users.
    However, in the proposed scheme one may use SIC instead of DPC for Charlie, as is done in the layered-SIC scheme for the MIMO wiretap problem.
    Alternatively, by using lower-triangular matrices instead of upper-triangular ones in \eqref{eq:EffMatGSVD} (which corresponds to switching roles between Bob and Charlie in the construction of the scheme), 
    one can use SIC for Bob and DPC for Charlie.
    This phenomenon was also observed by Liu \etal~\cite{ConfidentialMIMO_BC}.
    Unfortunately, this scheme does not allow, in general, to avoid DPC for both of the users.
 \end{remark}

 \begin{remark}[Other choices of precoding matrices]
    In \secref{ss:SpecialCasesOfTriangularizations}, different choices of $\bV_A$ were proposed for the MIMO wiretap problem: diagonalizing either $\bT_B$ or $\bT_C$, 
    which corresponds to avoiding SIC by Bob or guaranteeing strong secrecy, respectively; 
    or, by balancing all the SNRs of the sub-channels to Bob, which allows using the same codebook over all sub-channels and avoiding bit-loading / rate allocation.
    The analog in the case of confidential broadcast can be achieved by applying block diagonal unitary operations, in addition to the matrix $\bV_A$ that is dictated by the GSVD, 
    where the blocks correspond to the sub-channels that are allocated to Bob and to Charlie, of dimensions $\Lb \times \Lb$ and $\Lc \times \Lc$, respectively.
    However, whereas we can avoid SIC and DPC at Bob's end in the layered confidential broadcast scheme by diagonalizing his channel, 
    we cannot achieve this result for both Charlie and Bob simultaneously, as DPC needs to be employed for at least one of the users.
 \end{remark}

%%%%%%%%%%%%%%%%%%%%%%%%%%%%%%%%%%%%%%%%%%%%%%%%%%%%%%%%%%%%%%%%%%%%%%%%%%%%%%%%%%%%%%%%%%%%%%%%%%%%%%%%%%%%%%%%%%%%%%%%%

\section{Discussion: From Random Ensembles to Specific Codes}
\label{s:discussion}

In this work, we have demonstrated how scalar codes can be used for some MIMO secrecy scenarios.
Throughout the work, we have assumed that these scalar codes are taken from a random Gaussian ensemble, 
suitable in an appropriate sense (with or without secrecy constraints, with or without side information). 
One may be interested in a stronger result, where \emph{any} scalar codes that are good in the appropriate sense can be used, 
without worrying about the way they were created. 
Further, 
it is desirable to construct MIMO secrecy schemes using \emph{any standard} (non-secrecy) scalar codes that are good for communication over the (non-secrecy) AWGN channel.
To that end, 
one may hope to combine the approach of the current work with procedures that construct scalar wiretap codes from non-secrecy ones, 
such as \cite{WiretapCodeExpanderISIT2014} 
(which is based upon similar techniques for discrete wiretap channels proposed in \cite{SemanticWiretapCRYPTO,WiretapCodesFromOrdinaryCodesISIT2010}). 
Unfortunately, as we report in \cite{STUD:Wiretap:ITW2014}, there are some obstacles.

Surprisingly, the problem lies already in the use of scalar codes for MIMO communications without secrecy constraints. 
Recall the V-BLAST/GDFE schemes presented in \secref{ss:MIMO_P2P:SIC} and depicted in \figref{fig:VBLAST}.
Such schemes are widely accepted in the literature as capacity achieving, 
without proposing any treatment or analysis for specific codes.
In practice, such schemes are used in conjunction with arbitrary scalar codebooks, \eg, 
one-dimensional constellations with some error-correction code~\cite{PalomarJiang}; 
however, the combination does not necessarily approach capacity even if the individual codes do.
Indeed, for some specific channel matrices, the scheme might perform very poorly. 
To see this, consider \eqref{eq:successive}. % \eqref{eq:rx}.
This is a multiple-access channel (MAC) from the inputs $\tilde x_1,\ldots,\tilde x_i$ to the output $y'_{B;i}$. 
The SIC decoder treating all inputs as noise is equivalent to a stage of a successive-decoding procedure for the MAC. 
For the MAC, in turn, not any collection of good AWGN codes achieves capacity (see, \eg, \cite{BaccelliElGamalTse_MAC}). 
For example, assume that a MAC is given by 
\begin{align}
y_B = x_1 + x_2 + z.
\end{align}
Now further assume that the two codebooks are nested lattices. 
In that case (up to shaping), any possible point of $x_1+x_2$ is also a point of the higher-rate code, 
thus one codebook cannot be decoded without the other. 
The problem is not restricted to integer coefficient ratios but affects performance for coefficients close to any ``simple'' ratio; 
see, \eg, \mbox{\cite[Section III]{OrdentlichErez_LatticeAlignment}}. 

Returning back to the multi-stream schemes for the MIMO wiretap setup of \secref{s:schemes}, 
the decoder of Bob will also incur the same difficulty discussed above when generalizing to arbitrary scalar codes.
Furthermore, the same issue arises in our secrecy analyses (except when Eve's channel is orthogonalized, as in \secref{s:SVD@Eve}):
We successively provide Eve with previous messages as a ``genie'' side information. 
As a result the proof hinges on 
Eve's disability to perform a successive decoding process in the presence of interference from yet undecoded messages.
Here also this interference is taken to be Gaussian and alignment might help Eve.

To conclude, of the two ingredients needed for adjusting \emph{any} codes that are good for communication over scalar AWGN channels to the MIMO wiretap channel, the secrecy part can be treated by the procedure of~\cite{WiretapCodeExpanderISIT2014}. %\cite{TyagiVardy:Full}.
The remaining problem is similar to the one in SIC without secrecy constraints.
Indeed, obtaining good scalar Gaussian codes that approach capacity under SIC (without secrecy) from arbitrary scalar Gaussian codes remains an interesting open problem.

%%%%%%%%%%%%%%%%%%%%%%%%%%%%%%%%%%%%%%%%%%%%%%%%%%%%%%%%%%%%%%%%%%%%%%%%%%%%%%%%%%%%%%%%%%%%%%%%%%%%%%%%%%%%%%%%%%%%%%%%%

\appendices

%%%%%%%%%%%%%%%%%%%%%%%%%%%%%%%%%%%%%%%%%%%%%%%%%%%%%%%%%%%%%%%%%%%%%%%%%%%%%%%%%%%%%%%%%%%%%%%%%%%%%%%%%%%%%%%%%%%%%%%%%%%%%%%%%

\section{Proof of \lemref{lem:positive}}
\label{app:KochmanLemma}

The following proposition will be used in the proof of \lemref{lem:positive}.

\begin{prop}
\label{lem:GEV_GSV}
    Let $\bA_1$ and $\bA_2$ be $m_1 \times n$ and $m_2 \times n$ full-rank matrices, respectively, 
    where $m_1 \geq n$ and $m_2 \geq n$.
    Consider the generalized eigenvalue (GEV) problem: 
    \begin{align} 
        \bA_1^\dagger \bA_1 \by = \lambda \bA_2^\dagger \bA_2 \by \,.
    \end{align}
    Then, the generalized eigenvalues of $(\bA_1^\dagger \bA_1, \bA_2^\dagger \bA_2)$, $\{\lambda_i\}$, 
    are the GSVs of $(\bA_1,\bA_2)$, $\{\mu_i\}$, and the generalized eigenvectors are the corresponding columns of 
    \begin{align} 
        \bY = \bX^{-\dagger} . 
    \end{align}  
    Furthermore, the differential of the GEV $\lambda$ in terms of the differentials of $\bA_1^\dagger \bA_1$ and of $\bA_2^\dagger \bA_2$ is given by
    \begin{align}
    \label{eq:dlambda:general}
        d\lambda = \frac{\by^\dagger \Bigl(d( \bA_1^\dagger \bA_1)-\lambda d( \bA_2^\dagger \bA_2)\Bigr) \by}{\by^\dagger  \bA_1^\dagger \bA_1   \by  } \:\:. 
    \end{align}
\end{prop}

\begin{IEEEproof}
    The first part of the proposition easily follows from 
    \begin{align} 
        \bG_B^\dagger \bG_B \bY &= \bX \bL_B^2 \,,
     \\ \bG_E^\dagger \bG_E \bY &= \bX \bL_E^2 \,.
    \end{align}
    The proof of the differential identity \eqref{eq:dlambda:general} can be derived by standard eigenvalue perturbation analysis; see, \eg, \cite{GEV_derivative}.
\end{IEEEproof}

\vspace{.5\baselineskip}
Consider now the diagonal variant of the GSVD of $\bG_B = \bG(\bH_B, \bK)$ and $\bG_E = \bG(\bH_E, \bK)$ \eqref{eq:GSVD:diagonal}:
\begin{subequations}
\label{eq:appendix:GSVD}
\noeqref{eq:appendix:GSVD:Bob,eq:appendix:GSVD:Charlie}
\begin{align}
    \bG_B &= \bU_B \bL_B \bX^\dagger ,
\label{eq:appendix:GSVD:Bob}
 \\ \bG_E &= \bU_E \bL_E \bX^\dagger ,
\label{eq:appendix:GSVD:Charlie}
\end{align}
\end{subequations}
and denote the squared GSV vector by ${\bm\lambda}$, \ie, the vector whose entries satisfy:
\begin{align}
  \lambda_i \triangleq \mu^2_i \,.
\end{align}
Note further that $0 < \mu_i, \lambda_i < \infty$, since $\bG_B$ and $\bG_E$ are of full rank [recall \eqref{eq:G_B}].

Following \eqref{eq:I_S_b_c}, 
the MI difference in terms of $\{\lambda_i\}$ is equal to 
\begin{align}  
    I_S(\bH_B,\bH_E,\bK) 
 = \sum \log \lambda_i \,.
\end{align}

By applying the result of \propref{lem:GEV_GSV} to the effective channel matrices of~\eqref{eq:appendix:GSVD}, we obtain the following lemma.

\begin{lemma}
    The differential of the GSV $\lambda_i$ ($i = 1, \ldots, \Na$), in terms of the differential of the covariance matrix $\bK$, is given~by 
    \begin{align} 
        e_i^2 d\lambda_i = (\lambda_i-1) \by_i^\dagger \bB^{-1} (d\bK) \bB^{-\dagger} \by_i \,,
    \end{align} 
    where $\bB = \bK^{1/2}$, $\be$ is the diagonal of $\bL_E$, and $\by_i$ is the corresponding generalized eigenvector corresponding to $\lambda_i$.
\end{lemma}

\begin{IEEEproof}
    Perturbing $\bK$ results in the following differentials of 
    $\bG_B^\dagger \bG_B$ and $\bG_E^\dagger \bG_E$~\eqref{eq:appendix:GSVD} :
    \begin{subequations}
    \label{eq:dG_as_dK}
    \noeqref{eq:dG_as_dK:Bob,eq:dG_as_dK:Charlie}
    \begin{align} 
        2 d( \bG_B^\dagger \bG_B) &= \bB^{-1} (d\bK) \bH_B^\dagger \bH_B \bB
        + \bB^{\dagger}  \bH_B^\dagger \bH_B (d\bK) \bB^{-\dagger}, 
    \ \ \ \ 
    \label{eq:dG_as_dK:Bob}
     \\ 2 d( \bG_E^\dagger \bG_E) &= \bB^{-1} (d\bK) \bH_E^\dagger \bH_E \bB
        + \bB^{\dagger}  \bH_E^\dagger \bH_E (d\bK) \bB^{-\dagger} .
    \ \ \ \ 
    \label{eq:dG_as_dK:Charlie}
    \end{align} 
    \end{subequations}

    Substituting \eqref{eq:dG_as_dK} in \eqref{eq:dlambda:general}, gives rise to 
    \begin{align} 
        2 e_i^2 d\lambda_i & = \by_i^\dagger \Bigl( \bB^{-1} (d\bK) (\bH_B^\dagger \bH_B  - \lambda_i \bH_E^\dagger \bH_E) \bB \\ & \ \ \ +  \bB^{\dagger}  (\bH_B^\dagger \bH_B  - \lambda_i \bH_E^\dagger \bH_E) (d\bK) \bB ^{-\dagger} \Bigr) \by_i  \\
        & =  \by_i^\dagger \Bigl( \bB^{-1} (d\bK) \bB^{-\dagger} \bB^{\dagger} (\bH_B^\dagger \bH_B  - \lambda_i \bH_E^\dagger \bH_E) \bB \\ & \ \ \ +  \bB^{\dagger}  (\bH_B^\dagger \bH_B  - \lambda_i \bH_E^\dagger \bH_E) \bB  \bB ^{-1}   (d\bK) \bB ^{-\dagger} \Bigr) \by_i \\
        & = 2 (\lambda_i-1) \by_i^\dagger \bB^{-1} (d\bK) \bB^{-\dagger}  \by_i \,,
    \end{align}
    as desired.
\end{IEEEproof}

\begin{corol}
If $d\bK$ is positive semidefinite, then the sign of $d\lambda_i$ equals the sign of $\lambda_i-1$.
\end{corol}

The result of \lemref{lem:positive} follows immediately from this corollary.

%%%%%%%%%%%%%%%%%%%%%%%%%%%%%%%%%%%%%%%%%%%%%%%%%%%%%%%%%%%%%%%%%%%%%%%%%%%%%%%%%%%%%%%%%%%%%%%%%%%%%%%%%%%%%%%%%%%%%%%%%%%%%%%%%

{

\renewcommand{\tbL}{\bD'}
\renewcommand{\tL}{D'}
\renewcommand{\tbG}{\bG'}
\renewcommand{\tbT}{\bT'}
\renewcommand{\tT}{T'}
\renewcommand{\tbU}{\bU'}

\section{Truncation of Generalized Singular Values}
\label{app:GSVclipping}

Apply the triangular variant of the GSVD ~\eqref{eq:GSVD:Triangular} to the matrices $\bG_B = \bG(\bH_B, \bK)$ and $\bG_E = \bG(\bH_E, \bK)$, as in \eqref{eq:G_B} and \eqref{eq:EffMatGSVD}:
\begin{subequations}
\label{eq:app:GSVD:overI}
\noeqref{eq:app:GSVD:overI:B,eq:app:GSVD:overI:C}
\begin{align}
    \bG_B &\triangleq
    \begin{pmatrix}
             \bH_B \bK^{1/2} 
           \\ \bI 
    \end{pmatrix}
    = 
    \bU_B \bL_B \bT \bV_A^\dagger \,,
\label{eq:app:GSVD:overI:B}
\\ 
    \bG_E &\triangleq
    \begin{pmatrix}
             \bH_E \bK^{1/2} 
           \\ \bI 
    \end{pmatrix}
    = 
    \bU_E \bL_E \bT \bV_A^\dagger \,.
\label{eq:app:GSVD:overI:C}
\end{align}
\end{subequations}
Using any unitary matrix $\bQ$ instead of $\bI$ in the definition of $\bG_B$ and $\bG_E$, 
has no effect on the resulting matrices $\bV_A$, $\bT$, $\bL_B$ and $\bL_E$:
\begin{align}
    \begin{pmatrix}
             \bH_B \bK^{1/2} 
           \\ \bQ
    \end{pmatrix}
    = 
    \bU^\bQ_B \bL_B \bT \bV_A^\dagger \,,
\\ 
    \begin{pmatrix}
             \bH_E \bK^{1/2} 
           \\ \bQ
    \end{pmatrix}
    = 
    \bU^\bQ_E \bL_E \bT \bV_A^\dagger \,.
\end{align}
Furthermore, the upper-left $\Nb \times \Na$ and $\Nb \times \Ne$ of the resulting left unitary matrices $\bU^\bQ_B$ and $\bU^\bQ_E$, respectively, are equal to those of $\bU_B$ and $\bU_E$ of \eqref{eq:app:GSVD:overI}.

Using the last observation with $\bQ = \bV_A^\dagger$ and $\eqref{eq:app:GSVD:overI}$ for the matrices
\begin{align}
    \bG^\bV_B &\triangleq
    \begin{pmatrix}
             \bH_B \bK^{1/2} \bV_A
           \\ \bI 
    \end{pmatrix}
    =
    \begin{pmatrix}
             \bH_B \bK^{1/2}
           \\ \bV_A^\dagger
    \end{pmatrix}
    \bV_A
    \,,
\\ 
    \bG^\bV_E &\triangleq
    \begin{pmatrix}
             \bH_E \bK^{1/2} \bV_A
           \\ \bI 
    \end{pmatrix}
    =
    \begin{pmatrix}
             \bH_E \bK^{1/2}
           \\ \bV_A^\dagger
    \end{pmatrix}
    \bV_A
    \,, 
\end{align}
gives rise to the GSVD of $\bG^\bV_B$ and $\bG^\bV_E$:
\begin{subequations}
\label{eq:app:G'}
\noeqref{eq:app:G':B,eq:app:G':C}
\begin{align}
    \bG^\bV_B &\triangleq
    \bU^\bV_B \bL_B \bT \,,
\label{eq:app:G':B}
\\ 
    \bG^\bV_E &\triangleq
    \bU^\bV_E \bL_E \bT \,,
\label{eq:app:G':C}
\end{align}
\end{subequations}
where $\bU^\bV_B$ and $\bU^\bV_E$ are unitary (and their $\Nb \times \Na$ and $\Ne \times \Na$ upper-left sub-matrices are equal to those of $\bU_B$ and $\bU_E$, respectively).

That is, the GSVD of $\bG^\bV_B$ and $\bG^\bV_E$ is achieved by applying a QR decomposition to each of them.

The representation in \eqref{eq:app:G'} allows us to incorporate a truncation operation:
\begin{subequations}
\label{eq:truncatedGSVD}
\noeqref{eq:truncatedGSVD:Bob,eq:truncatedGSVD:Charlie}
\begin{align}
    \tbG_B &\triangleq
    \begin{pmatrix}
             \bH_B \bK^{1/2} \bV_A \bI_B
           \\ \bI
    \end{pmatrix}
 \\* &= 
    \tbU_B \tbL_B \tbT 
\label{eq:truncatedGSVD:Bob}
\\ 
    \tbG_E &\triangleq
    \begin{pmatrix}
             \bH_E \bK^{1/2} \bV_A \bI_B
           \\ \bI 
    \end{pmatrix}
 \\* &= 
    \tbU_E \tbL_E \tbT 
    \,,
\label{eq:truncatedGSVD:Charlie}
\end{align}
\end{subequations}
where $\tbU_B$ and $\tbU_E$ are unitary having the same first $\Lb$ columns as $\bU^\bV_B$ and $\bU^\bV_E$, respectively;
$\tbT$, $\tbL_B$ and $\tbL_E$ have the same first $\Lb$ columns as $\bT$, $\bL_B$ and $\bL_E$, respectively, 
whereas the remaining $\Le = \Na - \Lb$ columns are all zero except for the diagonal elements, which are equal to 1:
\begin{align}
    \tL_{B;i,j} &= \tL_{E;i,j} 
    = \tT_{i,j} = 1, 
    & i &= j, j > \Lb \,;
 \\ \tL_{B;i,j} &= \tL_{E;i,j} 
    = \tT_{i,j} = 0, 
    & i &\neq j, j > \Lb
    \,.
\end{align}
The latter is easily seen by noting that the QR decomposition carries out a Gram--Schmidt process over the columns of the decomposed matrices, and hence the first $\Lb$ columns 
remain the same after applying $\bI_B$, whereas the structure of the remaining columns is trivial due to the nullification of the last $\Le$ columns of $\bH_B \bK^{1/2} \bV_A$.

We note that \eqref{eq:truncatedGSVD} is the GSVD of $\tbG_B$ and $\tbG_E$ up to the normalization property \eqref{eq:GSVD:normalization}, 
which has no effect on the GSVs and can be achieved by a multiplication by an $\Na \times \Na$ diagonal matrix with its first $\Lb$ entries equal to 1 and the remaining entries~--- to $1/\sqrt{2}$.

The desired result is established by noting that $\bK^{1/2} = \obK^{1/2} \bV_A \bI_B$, and that the first $\Lb$ GSVs of $(\tbG_B, \tbG_E)$ are equal to the first $\Lb$ GSVs of $(\bG_B, \bG_E)$ (the GSVs that are greater than 1) and the remaining GSVs of $(\tbG_B, \tbG_E)$ are equal to~1.
}

%%%%%%%%%%%%%%%%%%%%%%%%%%%%%%%%%%%%%%%%%%%%%%%%%%%%%%%%%%%%%%%%%%%%%%%%%%%%%%%%%%%%%%%%%%%%%%%%%%%%%%%%%%%%%%%%%%%%%%%%%

% \section{Superposition Coding for the Wiretap Channel}
\section{Proof of \propref{thm:superposition}}
\label{app:SIC_optimality}

In this appendix, with a slight abuse of notation, we denote by boldface letters 
$n$-length sequences, with $n$ being the block length (in contrast to the other parts of the paper, where 
boldface letters denote spatial vectors).

\begin{proof}[Proof of \propref{thm:superposition}]
Denote 
  \begin{align}
  \label{eq:Superposition:tRk}
    \tR_k \triangleq I(\rvx_k; \rvy_E | \rvx_{k+1}^\Na) - \eps
    .
  \end{align}

  The codebooks are generated sequentially, from last ($k = \Na$) to first ($k = 1$), as follows.
  For $k = \Na$, construct the codebook $\cC_\Na$ of $2^{n \left( R_\Na + \tR_\Na \right)}$ codewords, 
  that are generated independently with i.i.d.\ entries with respect to $p \left( \rvx_\Na \right)$.
  For $k \in \{1, \ldots, \Na - 1\}$, for each (already generated) codeword set 
  \mbox{$(\tbx_{k+1}, \ldots, \tbx_\Na) \in \cC_{k+1} \times \cdots \times \cC_\Na$}, 
  generate a codebook of $2^{n \left( R_k + \tR_k \right)}$ codewords with respect to $\prod_{i=1}^n p \left( \rvx_k \middle| \tx_{k+1}(i), \ldots, \tx_\Na(i) \right)$, where $\tx_\ell(i)$ is the $i$-the letter of the codeword $\tbx_\ell$.
  Within each codebook, each codeword is assigned a unique index pair $(\rvm_k,\rvf_k)$ where  
  $\rvm_k \in \{ 1,2,\ldots, 2^{n R_k} \}$ and 
  $\rvf_k \in \{ 1,2,\ldots, 2^{n \tR_k} \}$. Each codeword is selected according to the secret message $\rvm_k$ and a fictitious message $\rvf_k$ drawn uniformly over its range.  The transmitted codeword is therefore $\bx = \vphi \left( \tbx_1(\rvm_1, \rvf_1), \ldots, \tbx_\Na(\rvm_\Na, \rvf_\Na) \right)$. 
  Bob's decoding is based on successive decoding starting from the last message ($k = \Na$) and proceeding to the first ($k = 1$).
  \ver{\\ \indent}{}Since
  \ver{\begin{subequations}
  \noeqref{eq:Superposition:Rk+tRk,eq:Superposition:Rk+tRk:UB}}{}
  \begin{align}
    \label{eq:Superposition:Rk+tRk}
        R_k + \tR_k  &= I \left( \rvx_k; \rvy_B \middle| \rvx_{k+1}^\Na \right) - 2 \eps
   \ver{\\ &}{}<  I \left( \rvx_k; \rvy_B \middle| \rvx_{k+1}^\Na \right) , \quad \:
    \ver{\label{eq:Superposition:Rk+tRk:UB}}{}
  \end{align}
  \ver{\end{subequations}}{}
  the decoding of each combined message $(\rvm_k,\rvf_k)$ succeeds with arbitrarily high probability, as $n\rightarrow\infty$.

  In order to satisfy the secrecy constraint, the following condition must hold, for any $\teps > 0$ and large enough $n$:
  \begin{align}
    \frac{1}{n} H \left( \rvm_1,\ldots, \rvm_\Na \middle| \by_E, \cC \right) \ge \frac{1}{n}H(\rvm_1,\ldots, \rvm_\Na) - \teps \,,
  \end{align}
  where $\cC = \{\cC_1, \ldots, \cC_\Na\}$ denotes the overall collection of the $\Na$ codebooks. 

  It suffices to show that for any $\eps' > 0$, and large enough $n$, 
  \begin{align}
  \label{eq:cond-lb}
    \frac{1}{n} H(\rvm_k | \by_E, \rvm_{k+1}^\Na,\cC) \ge \frac{1}{n} H(\rvm_k) - \eps'
  \end{align}
  is satisfied for each $k$. \ver{\\ \indent}{}Note that 
  \begin{align}
      &H \left( \rvm_k \middle| \by_E, \rvm_{k+1}^\Na, \cC \right) \ge H \left( \rvm_k \middle| \by_E, \tbx_{k+1}^\Na, \cC \right) 
   \\ &= H \left( \rvm_k, \tbx_k \middle| \by_E, \tbx_{k+1}^\Na, \cC \right) 
       - H \left( \tbx_k \middle| \rvm_k, \by_E, \tbx_{k+1}^\Na, \cC \right) 
   \\ &= H \left( \tbx_k \middle| \by_E, \tbx_{k+1}^\Na, \cC \right) 
       - H \left( \rvf_k \middle| \rvm_k, \by_E, \tbx_{k+1}^\Na, \cC \right) .
  \end{align}
  Due to \eqref{eq:Superposition:tRk}, 
  in our construction the eavesdropper can decode $\rvf_k$ with probability going to 1, given 
  $\left( \rvm_k, \by_E, \tbx_{k+1}^\Na, \cC \right)$, and hence the second term \ver{is vanishingly small}{vanishes to zero}. 
  Thus, we are left with
  \begin{align}
  \label{eq:ent-diff}
    \begin{aligned}
        &H \left( \rvm_k \middle| \by_E, \rvm_{k+1}^\Na, \cC \right) \ge  H \left( \tbx_k \middle| \by_E, \tbx_{k+1}^\Na, \cC \right) 
        - n \eps'_n
    \\ &= H \left( \tbx_1^k \middle| \tbx_{k+1}^\Na, \by_E, \cC \right)
        - H \left( \tbx_1^{k-1} \middle| \tbx_k^\Na, \by_E, \cC \right) -n \eps'_n \,.
    \end{aligned}
  \end{align}
  Since the two equivocations are the same quantity up to an index shift, 
  it suffices to show that for $\del_1 > 0$ and $\del_2 > 0$ that vanish with $\eps$ and large enough $n$, 
\begin{subequations}
  \noeqref{eq:anal0}
  \begin{align}
    & \sum_{\ell=1}^k \left[ I\left( \rvx_\ell; \rvy_B \middle| \rvx_{\ell+1}^\Na \right) 
        - I \left( \rvx_\ell; \rvy_E \middle| \rvx_{\ell+1}^\Na \right) \right] - \del_1 
  \label{eq:anal0}
   \\ &\le \frac{1}{n} H \left( \tbx_1^k \middle| \tbx_{k+1}^\Na, \by_E, \cC \right)        
  \label{eq:anal1}
   \\ &\le \sum_{\ell=1}^{k} I \left( \rvx_\ell; \rvy_B \middle| \rvx_{\ell+1}^\Na \right) 
        - I \left( \rvx_\ell; \rvy_E \middle| \rvx_{\ell+1}^\Na \right) + \del_2
  \label{eq:anal2}
  \,.
  \end{align}
\end{subequations}

  To establish~\eqref{eq:anal1} we use the fact that the sequences $\tbx_\ell$ are selected independently given $\tbx_{\ell+1}^\Na$, 
  so that, for large enough $n$, the following chain of inequalities holds 
  \begin{subequations}
  \noeqref{eq:Superposition:equivocation:LB:a,eq:Superposition:equivocation:LB:b,eq:Superposition:equivocation:LB:c,eq:Superposition:equivocation:LB:iid-bnd,eq:Superposition:equivocation:LB:Rk+tRk}
  \begin{align}
      &H \left( \tbx_1^k \middle| \tbx_{k+1}^\Na, \by_E, \cC \right)
  \label{eq:Superposition:equivocation:LB:a}
   \\ &= H \left( \tbx_1^k \middle| \tbx_{k+1}^\Na, \cC \right) 
        - I \left( \tbx_1^k ; \by_E \middle| \tbx_{k+1}^\Na, \cC \right)
  \label{eq:Superposition:equivocation:LB:b}
   \\ &=\sum_{\ell=1}^k  \left[ H \left( \tbx_\ell \middle| \tbx_{\ell+1}^\Na, \cC \right) 
        - I \left( \tbx_\ell ; \by_E \middle| \tbx_{\ell+1}^\Na, \cC \right) \right]
  \label{eq:Superposition:equivocation:LB:c}
   \\ &= \sum_{\ell=1}^k  \Big[ n I \left( \rvx_\ell ; \rvy_B \middle| \rvx_{\ell+1}^\Na \right) - 2 \eps
   - I \left( \tbx_\ell ; \by_E \middle| \tbx_{\ell+1}^\Na, \cC \right) \Big] \ \ \ \
  \label{eq:Superposition:equivocation:LB:Rk+tRk}
   \\ &\ge n \sum_{\ell=1}^k \left[ I \left( \rvx_\ell ; \rvy_B \middle| \rvx_{\ell+1}^\Na \right)
       - I \left( \rvx_\ell ; \rvy_E \middle| \rvx_{\ell+1}^\Na \right) - 3 \eps \right]  
  \label{eq:Superposition:equivocation:LB:iid-bnd}
      ,
  \end{align}
  \end{subequations}
  where \eqref{eq:Superposition:equivocation:LB:Rk+tRk} follows from \eqref{eq:Superposition:Rk+tRk}, 
  and to establish \eqref{eq:Superposition:equivocation:LB:iid-bnd} we use the fact that the channel is memoryless along with standard typicality arguments \cite{CoverBook2Edition}.

  To establish \eqref{eq:anal2}, we use \cite[Lemma~1]{ChiaElGamal2012_3RxBCwithCommonAndConfidential}, by substituting:
  \begin{align}
      \bullet \, S &= \sum_{\ell=1}^k \left( R_\ell + \tR_\ell \right) &\, & \bullet \rvu  = \rvx_{k+1}^\Na
   \\ \bullet \,\: \rvv &= \rvx_1^k &\, & \bullet \rvz = \rvy_E
   \\ \bullet \, L &\defeq (\rvm_1^k, \rvf_1^k) \in [1,2^{nS}]
  \end{align}
  The conditions for the lemma hold since 
  \begin{align}
    H \left( \rvx_1^k \middle| \rvx_{k+1}^\Na, \rvy_E,\cC \right) &= H\left( L \middle| \rvx_{k+1}^\Na, \rvy_E,\cC \right) ,
  \end{align}
  and
  \begin{subequations}
  \noeqref{eq:ChiaElGamal:1,eq:ChiaElGamal:2,eq:ChiaElGamal:3}
  \begin{align}
      S &= \sum_{\ell=1}^k \left( R_\ell + \tR_\ell \right)
  \label{eq:ChiaElGamal:1}
   \\ &= \left[ \sum_{\ell=1}^k I \left( \rvx_\ell ; \rvy_B \middle| \rvx_{\ell+1}^\Na \right) \right] - 2\eps
  \label{eq:ChiaElGamal:2}
   \\ &> \left[ \sum_{\ell=1}^k I \left( \rvx_\ell ; \rvy_E \middle| \rvx_{\ell+1}^\Na \right) \right] + \delta
  \label{eq:ChiaElGamal:positive_rates}
   \\ &= I \left( \rvx_1^k ; \rvy_E \middle| \rvx_{k+1}^{\Na} \right) + \del \,,
  \label{eq:ChiaElGamal:3}
  \end{align}
  \end{subequations}
  where \eqref{eq:ChiaElGamal:positive_rates} follows from the fact that the communication rate $R_\ell$ of each sub-channel must be positive 
  (and $\eps$ and $\del$ are small enough, and $n$ is sufficiently large), else it is not used.
  Since we have proved \eqref{eq:anal1} and \eqref{eq:anal2}, the secrecy analysis is now complete.
\end{proof}

\begin{remark}
    For the special case of mutually independent $\left( \rvx_1, \ldots, \rvx_\Na \right)$, 
    there is no need to generate a different codebook $\cC_k$ for each selection of preceding codewords $\left( \tbx_{k+1}, \ldots, \tbx_\Na \right)$, and the same codebook can be applied regardless of the other codewords.
\end{remark}

%%%%%%%%%%%%%%%%%%%%%%%%%%%%%%%%%%%%%%%%%%%%%%%%%%%%%%%%%%%%%%%%%%%%%%%%%%%%%%%%%%%%%%%%%%%%%%%%%%%%%%%%%%%%%%%%%%%%%%%%%

\section*{Acknowledgment}

The authors thank Ziv Goldfeld for proposing to extend the result of \propref{thm:superposition} from independent codes $x_1,\ldots,x_\Na$ to dependent ones, and Ronit Bustin for helpful discussions
and for pointing their attention to the work of Baccelli \etal~\cite{BaccelliElGamalTse_MAC}.

%%%%%%%%%%%%%%%%%%%%%%%%%%%%%%%%%%%%%%%%%%%%%%%%%%%%%%%%%%%%%%%%%%%%%%%%%%%%%%%%%%%%%%%%%%%%%%%%%%%%%%%%%%%%%%%%%%%%%%%%%

% Generated by IEEEtran.bst, version: 1.12 (2007/01/11)

% \bibliography{toly}
% \bibliographystyle{IEEEtran}

%%%%%%%%%%%%%%%%%%%%%%%%%%%%%%%%%%%%%%%%%%%%%%%%%%%%%%%%%%%%%%%%%%%%%%%%%%%%%%%%%%%%%%%%%%%%%%%%%%%%%%%%%%%%%%%%%%%%%%%%%
\end{document}